\documentclass[letterpaper, 10 pt, conference]{ieeeconf}  

\IEEEoverridecommandlockouts                              

\usepackage{algorithmic}
\let\labelindent\relax
\usepackage{enumitem}
\usepackage{amssymb}

\usepackage{amsthm}
\usepackage{mathtools}
\usepackage{bm} 
\usepackage{graphicx} 
\usepackage{algorithm}
\usepackage{bbm} 
\usepackage{hyperref}
\usepackage{mathtools}
\usepackage{chngcntr}

\DeclareMathOperator*{\argmin}{arg\,min} 

\DeclareMathOperator{\erf}{erf}
\DeclareMathOperator{\erfc}{erfc}

\newcommand{\Var}[1]{\mathrm{Var}\left(#1\right)} 

\newcommand{\abs}[1]{\left| #1 \right|}
\newcommand{\norm}[1]{\left\lVert #1 \right\rVert}

\newcommand{\trace}[1]{\text{tr} \left( #1 \right)}
\newcommand{\erff}[1]{\erf \left( #1 \right)}
\newcommand{\erfcf}[1]{\erfc \left( #1 \right)}
\newcommand{\expf}[1]{\exp \left( #1 \right)}

\newtheorem{theorem}{Theorem}

\newtheorem{lemma}[theorem]{Lemma}
\newtheorem{proposition}[theorem]{Proposition}

\theoremstyle{definition}
\newtheorem{definition}{Definition}

\newtheorem*{remark}{Remark}

\mathtoolsset{showonlyrefs=true} 
\usepackage{xcolor}
\newcommand{\ssnote}{\textcolor{black}}

\usepackage{caption}
\usepackage{subcaption}

\title{\LARGE \bf
Learning-Based Adaptive Control for Stochastic Linear Systems with Input Constraints
}

\author{Seth Siriya, Jingge Zhu, Dragan Ne\v{s}i\'{c}, and Ye Pu 
\thanks{The work of Seth Siriya was supported by the Australian Government Research Training Program (RTP) Scholarship. The work of Dragan Ne\v{s}i\'{c} was supported by the Australian Research Council through the Discovery Project under Grant
DP210102600. The work of Ye Pu was supported by the Australian Research Council under Project DE220101527.}
\thanks{The authors are with the Department of Electrical and Electronic Engineering, University of Melbourne, Parkville, VIC 3010, Australia (e-mail: {\tt\small \href{mailto:ssiriya@student.unimelb.edu.au}{ ssiriya@student.unimelb.edu.au}, \{\href{mailto:jingge.zhu@unimelb.edu.au}{jingge.zhu}, \href{mailto:dnesic@unimelb.edu.au}{dnesic}, \href{mailto:ye.pu@unimelb.edu.au}{ye.pu}\}@unimelb.edu.au}).}
}

\begin{document}

\maketitle
\thispagestyle{empty}
\pagestyle{empty}

\begin{abstract}
We propose a certainty-equivalence scheme for adaptive control of scalar linear systems subject to additive, i.i.d. Gaussian disturbances and bounded control input constraints, without requiring prior knowledge of the bounds of the system parameters, nor the control direction. Assuming that the system is at-worst marginally stable, mean square boundedness of the closed-loop system states is proven. Lastly, numerical examples are presented to illustrate our results. 
\end{abstract}

\section{Introduction}
Adaptive control is useful for stabilizing dynamical systems with known model structure but unknown parameters. However, a major problem that arises when deploying controllers is actuator saturation, which, when unaccounted for, can result in failure to achieve stability. Moreover, in some systems, large, stochastic disturbances may occasionally perturb the system, which needs to be considered in controller design. This motivates the need to develop controllers which can stabilize systems with unknown parameters, whilst simultaneously handling both control input constraints and additive, unbounded, stochastic disturbances.

\ssnote{In recent years, there has been renewed interest in discrete-time (DT) stochastic adaptive control. One reason is the recent successes in online model-based reinforcement learning --- especially for the online linear quadratic regulation (LQR) task, where the goal is to apply controls on an unknown linear stochastic system to minimize regret with respect to the optimal LQR controller in a single trajectory (e.g. see \cite{lale2022reinforcement, kargin2022thompson, simchowitz2020naive}). However, input constraints are not considered in these works. Recent results in \cite{li2021safe} address state and input constraints, but assume bounded disturbances, and require an \textit{a priori} known controller guaranteeing stability and constraint satisfaction. DT stochastic extremum seeking (ES) results have shown promise, and recently been applied beyond steady-state input-output maps to stabilize open-loop unstable systems (see \cite{ radenkovic2016stochastic, radenkovic2018extremum2}), but they also do not consider input constraints. Looking back to classic results in stochastic adaptive control such as \cite{goodwin1981discrete, meyn1987new, guo1996self}, many challenges have been addressed, but stochastic stability results considering bounded control constraints with unbounded disturbances for marginally stable plants are missing, despite almost all real systems having actuator constraints.}

Beyond DT stochastic adaptive control, various works consider input-constrained linear systems. \ssnote{Although seemingly simple, the stability analysis of controllers for input-constrained, marginally stable, linear plants, with unbounded disturbances, is non-trivial, due to their nonlinear, stochastic, closed-loop dynamics. Results reporting mean square boundedness with arbitrary positive input constraints for known systems were not available until after 2012 in \cite{chatterjee2012mean} and \cite{mishra2017output}.}
The adaptive control of unknown, DT output-feedback linear systems subject to bounded control constraints and bounded disturbances is considered in \cite{annaswamy1995discrete}, \cite{feng2001robust}, \cite{chaoui2001adaptive}, and \cite{zhang2001adaptive}. These works derive deterministic guarantees on at least the boundedness of the output under various conditions, but require bounded disturbances.

Although adjacent settings have been considered, to the best of our knowledge, no works address the problem of adaptive control for DT linear systems subject to control input constraints and unbounded disturbances with proven stability guarantees. We move towards filling this gap by addressing this task specifically for \textit{scalar}, at-worst marginally stable linear systems, with additive i.i.d. zero-mean Gaussian disturbances. Our main contributions are twofold:

Firstly, we propose a certainty-equivalence (CE) adaptive control scheme comprised of an ordinary least squares-based plant parameter estimator component, in connection with an excited, saturated deadbeat controller based on the estimated parameters. Our controller is capable of satisfying any positive upper bound constraint on the magnitude of the control input. Moreover, it does not assume prior knowledge of bounds for the system parameters, nor does it require knowledge of the control direction --- that is, the sign of the control parameter, which is a common assumption in adaptive control.

Secondly, we establish the mean square boundedness of the closed-loop system when applying our proposed our control scheme to the system of interest. Despite the restricted problem setting, it is still non-trivial since saturated controls render the system nonlinear, and in general CE control does not stabilize nonlinear systems \cite{krstic1995nonlinear}. We overcome this difficulty by showing that our control scheme satisfies sufficient excitation conditions required to establish an upper bound on the probability that the parameter estimate lies outside a small ball around the true parameter, making use of results in non-asymptotic learning from \cite{simchowitz2018learning}. We subequently show that this upper bound decays sufficiently fast, allowing us to prove via a novel analysis that mean square boundedness holds. Typically, persistence of excitation is assumed in the control literature to establish parameter convergence, whereas we explicitly demonstrate satisfaction of excitation conditions, which is non-trivial even in the scalar case due to the nonlinear, stochastic nature of our system. Moreover, establishing mean square boundedness is difficult due to the unbounded nature of the disturbances and saturated controls, requiring specialized results from \cite{pemantle1999moment} in our analysis.

\paragraph*{Notation} Let $\mathbb{N}$ denote the set of natural numbers, and $\mathbb{N}_0 := \mathbb{N}\cup \{0\}$. Let $\mathbb{R}$ denote the real numbers, and $\mathbb{R}_{\geq 0} := [0,\infty)$. For $x \in \mathbb{R}$, we define $x^+:=\max(0,x)$. It has the properties 1) $x = x^+ - (-x)^+$ and 2) $x^+(-x)^+=0$ for all $x \in \mathbb{R}$. Let $\mathcal{S}^{d-1}$ denote the unit sphere in $\mathbb{R}^d$. For $r>0$, we define the saturation function $\sigma_r:\mathbb{R}\rightarrow\mathbb{R}$ by $\sigma_r(x) := x$ if $\abs{x} \leq r$, and $\sigma_r(x) := r x / \abs{x}$ if  $\abs{x} > r$. For a square matrix $A \in \mathbb{R}^{d \times d}$, let $\lambda_{\text{min}}(A)$ and $\lambda_{\text{max}}(A)$ denote the minimum and maximum eigenvalue of $A$ respectively. For symmetric matrices $A,B \in \mathbb{R}^{d \times d}$, we denote $ A \prec B$ ($\preceq$) if $A-B$ is negative definite (semi-definite). Let $\erff{\cdot}$ denote the error function, and $\erfcf{\cdot}$ denote the complementary error function. Consider a probability space $(\Omega,\mathcal{F}, P)$, and a random variable $X:\Omega \rightarrow \mathbb{R}$, sub-sigma-algebra $\mathcal{G}\subseteq \mathcal{F}$, and events $A,B \in \mathcal{F}$ defined on this space. Let $\mathbb{E}[\cdot]$ denote the expectation operator. We say $X | \mathcal{G}$ is $\Sigma^2$-sub-Gaussian if  $\mathbb{E}[e^{\lambda X} | \mathcal{G}] \leq e^{\Sigma^2 \lambda^2 /2}$ for all $\lambda \in \mathbb{R}$. For an event $A \in \mathcal{F}$, we define the indicator function $\bm{1}_A:\Omega \rightarrow \{0,1\}$ as $\bm{1}_A := 1$ on the event $A$, and $\bm{1}_A := 0$ on the event $A^C$. If $X$ takes values in $\mathbb{R}_{\geq 0}$, then $X \bm{1}_{A \cup B} = \max (X \bm{1}_{A}, X \bm{1}_B)$ holds. \ssnote{$\dagger$ denotes the Moore-Penrose inverse.}

\section{Problem Setup}
Consider the stochastic scalar linear system:
\begin{equation}\label{eqn:scalar-system}
    X_{t+1} = aX_t + bU_t + W_t, \ t \in \mathbb{N}_0, \ X_0 = x_0
\end{equation}
where the random sequences $(X_t)_{t\in \mathbb{N}_0}$, $(U_t)_{t\in \mathbb{N}_0}$ and $(W_t)_{t \in \mathbb{N}_0}$ are the states, controls and disturbances taking values in $\mathbb{R}$, $x_0 \in \mathbb{R}$ is the initial state, and $\theta_*=(a,b) \in \mathbb{R}^2$ are the system parameters. Throughout the paper, all random variables are defined on a probability space $(\Omega,\mathcal{F},P)$. Moreover, denote the $i$-th moment of the disturbance as $S_i := \mathbb{E}[\abs{W_t}^i]$ for $i \in \mathbb{N}$. We make the following assumptions on the system in \eqref{eqn:scalar-system}.

\begin{description}
    \item[A1.] The disturbance sequence $(W_t)_{t \in \mathbb{N}_0}$ is sampled $W_t \stackrel{\text{i.i.d.}}{\sim} \mathcal{N}(0,\Sigma_W^2)$ where $\Sigma_W^2 > 0$ is the variance;
    \item[A2.] The system parameters $(a,b)$ satisfy $\abs{a} \leq 1$ and $b \neq 0$.
\end{description}

\begin{remark}
    A1 is selected since Gaussian random variables practically model many types of disturbances due to their unbounded support, which can represent rare events that cause arbitrarily large disturbances in real systems. A2 ensures the existence of control policies with bounded control constraints that render the system mean square bounded, as proven in \cite{chatterjee2012mean}. Heuristically speaking, this is because A2 ensures global null-controllability in the deterministic setting, which is intuitively important because A1 can cause arbitrarily large jumps in the state.
\end{remark}

Our goal is to formulate an adaptive control policy $(\pi_t)_{t \in \mathbb{N}_0}$ such that $\pi_t$ is a mapping from current and past state and control input data $(X_0,\hdots,X_t,U_0,\hdots,U_{t-1})$ and a randomizaton term $V_t$ to $\mathbb{R}$ for $t \in \mathbb{N}_0$, where $(V_t)_{t\in\mathbb{N}_0}$ taking values in $\mathbb{R}$ is an i.i.d. random sequence. We allow for stochastic policies (i.e. dependence on $V_t$) to excite the system and facilitate parameter convergence. Moreover, we require as part of our design that $\pi_t$ does not depend on the system parameters $(a,b)$, and that the following requirements are satisfied on the closed-loop system with $U_t = \pi_t(X_0,\hdots,X_t,U_0,\hdots,U_{t-1},V_t)$ for $t \in \mathbb{N}_0$:
\begin{description}
    \item[G1.] The magnitude of the control input remains bounded by a desired constraint level: $|U_t| \leq U_{\text{max}}$ for all $t \in \mathbb{N}_{0}$, where $U_{\text{max}} > 0$ is a user specified constraint level;
    \item[G2.] For the stochastic process $(X_t)_{t\in \mathbb{N}_0}$, mean-square boundedness is achieved: $\exists e > 0 : \sup_{t \in \mathbb{N}_0} \mathbb{E}\left[ X_t^2 \right] \leq e$.
\end{description}
In practice, $U_{\text{max}}$ is chosen based on the maximum control input the actuator can provide.

\section{Method and Main Result}
The control strategy we employ to achieve G1 and G2 is summarized in Algorithm \ref{alg:ideal-controller-v2}.

\begin{algorithm}[H]
\caption{AICMSS (Adaptive Input-Constrained Mean Square Stabilization)}
\begin{algorithmic}[1]
    \STATE \textbf{Inputs:} $U_{\text{max}} > 0$, $C < U_{\text{max}}$, \ssnote{$\hat{a}_{\text{init}} \in \mathbb{R}$, $\hat{b}_{\text{init}} \in \mathbb{R} \backslash \{ 0 \}$}
    \STATE $D \leftarrow U_{\text{max}} - C $
    \STATE Measure $X_0$
    \FOR{$t = 0,1,\hdots$}
        \STATE Sample $V_t \stackrel{\text{i.i.d.}}{\sim} \text{Uniform}([-C,C])$
        \STATE Compute control $U_t$ following \eqref{eqn:control-policy-v2}
        \STATE Apply $U_t$ to \eqref{eqn:scalar-system}
        \STATE Measure $X_{t+1}$ from \eqref{eqn:scalar-system}
        \STATE \ssnote{\textbf{if} $t \geq 1$ \textbf{then} Compute $(\hat{a}_{t},\hat{b}_{t})$ following \eqref{eqn:parameter-estimate}}
    \ENDFOR
\end{algorithmic}\label{alg:ideal-controller-v2}
\end{algorithm}

We now describe our strategy in greater detail. The sequence of control inputs $(U_t)_{t \in \mathbb{N}_0}$ is given by the policy:
\begin{align}
    U_t &:=  \sigma_D\left(\ssnote{G_t} X_t\right) + V_t, \label{eqn:control-policy-v2}
\end{align}
for $t \in \mathbb{N}_0$, where \ssnote{$G_t := -\hat{a}_{\text{init}}/\hat{b}_{\text{init}}$ for $t \leq 1$ and $G_t := -\hat{a}_{t-1}/\hat{b}_{t-1}$ for $t \geq 2$ is the \textit{gain factor}}, and $V_t \stackrel{\text{i.i.d.}}{\sim}\text{Uniform}([-C,C])$ is an additive \textit{excitation term} \ssnote{independent of $W_t$}. Moreover, $\hat{\theta}_t = (\hat{a}_t,\hat{b}_t) \in \mathbb{R}^2$ is the parameter estimate at time $t \in \mathbb{N}$ obtained via least squares estimation:
\begin{align}
    \hat{\theta}_t &\ssnote{= (\bm{Z}_t^{\top} \bm{Z}_t)^{\dagger}\bm{Z}_t^{\top} \bm{X}_{t+1}} \\
    &\in \arg \min_{\theta \in \mathbb{R}^{2}} \sum_{s=1}^t \norm{\ssnote{X_{s+1}} - \theta^{\top} Z_s}_2^2, \label{eqn:parameter-estimate}
\end{align}
where $\ssnote{(Z_t)_{t \in \mathbb{N}_0}} \subseteq \mathbb{R}^2$ is the state-input data sequence:
\begin{equation}
    Z_t := \ssnote{(X_{t},U_{t}), \quad t \in \mathbb{N}_0}, \label{eqn:state-input-data}
\end{equation}
\ssnote{and $\bm{Z}_t=[Z_1,\hdots,Z_t]^{\top}$, $\bm{X}_{t+1} = [X_2,\hdots,X_{t+1}]^{\top} $}. The initial parameter estimate $\hat{\theta}_{\text{init}}=(\hat{a}_{\text{init}},\hat{b}_\text{init})$ is freely chosen by the designer in $\mathbb{R}\times\mathbb{R}\backslash \{0 \}$. Here, $C$ is a user-specified excitation constant satisfying $0 < C < U_{\text{max}}$ which determines the size of the excitation term, and $D = U_{\text{max}}-C$ determines the certainty-equivalence component of the control policy.

Under this control strategy, the states of the closed-loop system evolve as,
\begin{equation} 
    X_{t+1} = aX_t + b \bigg( \sigma_D\bigg(\ssnote{G_t} X_t\bigg) + V_t \bigg) + W_t. \label{eqn:closed-loop-system-v2}
\end{equation}

The intuition behind our control strategy is the following; we estimate the system parameters $(a,b)$ from past data in \eqref{eqn:parameter-estimate}, and use this estimate for certainty-equivalent control in \eqref{eqn:control-policy-v2}. The presence of the \ssnote{injected excitation} term $V_t$ is to ensure the data $Z_t$ is exciting so that convergence of parameter estimates $(\hat{a}_t,\hat{b}_t)$ occurs, which is a key component of our analysis later in Theorem \ref{thm:msb}. \ssnote{This is in contrast to $W_t$, which is an external system disturbance}.

The control input \eqref{eqn:control-policy-v2} always satisfies (G1) since for all $t \in \mathbb{N}_0$, $\abs{\sigma_D\left(\ssnote{G_t}X_t\right)}$ is bounded by $D$, $\abs{V_t}$ is bounded by $C$, and $U_{\text{max}} = C+D$. On the other hand, Theorem \ref{thm:msb} states that the closed-loop system under our control law following Algorithm \ref{alg:ideal-controller-v2} satisfies mean-square boundedness (G2). We now provide Theorem \ref{thm:msb}.
\begin{theorem} \label{thm:msb}
Supposing A1-\ssnote{A2} hold \ssnote{and $x_0 \in \mathbb{R}$}, the states $(X_t)_{t \in \mathbb{N}_0}$ corresponding to the closed-loop system \eqref{eqn:closed-loop-system-v2} under our control strategy in Algorithm  \ref{alg:ideal-controller-v2} satisfies 
\begin{equation}
    \exists e > 0: \sup_{t \in \mathbb{N}_0}\mathbb{E}\left[ X_t^2 \right] < e.
\end{equation}
\end{theorem}

\section{Proof of Main Result}
We start by providing the main ideas behind the proof of Theorem \ref{thm:msb}. We then state supporting lemmas and a sketch of their proofs, before providing the formal proof of Theorem \ref{thm:msb}.

\subsection{Proof Idea for Theorem \ref{thm:msb}}
 Firstly, let us define the block martingale small-ball (BMSB) condition.
\begin{definition} \label{def:bmsb}
(Martingale Small-Ball \cite[Definition 2.1]{simchowitz2018learning}) Given a process $(Z_t)_{t \geq 1}$ taking values in $\mathbb{R}^2$, we say that it satisfies the $(k,\Gamma_{\text{sb}},p)$-block martingale small-ball (BMSB) condition for \ssnote{$k \in \mathbb{N}$}, $\Gamma_{\text{sb}} \succ 0$, and \ssnote{$p > 0$}, if, for any $\zeta \in \mathcal{S}^{1}$ and $j \geq 0$, $\frac{1}{k} \sum_{i=1}^k P (\abs{\langle \zeta, Z_{j+i} \rangle} \geq \sqrt{\zeta^{\top} \Gamma_{\text{sb}} \zeta} \mid \mathcal{F}_j) \geq p$ holds. Here, $(\mathcal{F}_t)_{t \geq 1}$ is any filtration which $(\langle \zeta, Z_{t} \rangle)_{t \geq 1}$ is adapted to.
\end{definition}
This condition is related to the `excitability' of some random sequence $(Z_t)_{t \geq 1}$ --- intuitively, that is, given past observations of the sequence $Z_1,\hdots,Z_j$, how spread out is the conditional distribution of future observations. Result 1) in Lemma \ref{lemma:satisfy-bmsb} establishes that our state-input data sequence $(Z_t)_{t \in \mathbb{N}}$ satisfies the $(1,\Gamma_{\text{sb}},p)$-BMSB for some parameters $p,\Gamma_{\text{sb}}$, which in turn implies a high probability lower bound on $\lambda_{\text{min}}(\sum_{t=1}^i Z_t Z_t^{\top})$ holds \cite{simchowitz2018learning}. Moreover, 2) in Lemma \ref{lemma:satisfy-bmsb} provides a high-probability upper bound on $\lambda_{\text{max}}(\sum_{t=1}^i Z_t Z_t^{\top})$. These bounds are important for deriving a high probability upper bound on the parameter estimation error when applying least squares estimation to a general time-series with linear responses (Theorem 2.4 in \cite{simchowitz2018learning}). We provide the specialization to the case of covariates in $\mathbb{R}^2$ and responses in $\mathbb{R}$ in Proposition \ref{prop:estim-bound}. We subsequently rely on this result to provide Lemma \ref{lemma:parameter-estimate-bound}, which gives an upper bound on the probability that the parameter estimate $\hat{\theta}_t$ lies outside a ball of size $d>0$ centered at the true parameter $\theta_*$ for sufficiently small $d$ for $t \in \mathbb{N}$.
Finally, Lemma \ref{lemma:msb-ce} says that so long as the aforementioned probability decays sufficiently fast, then our certainty-equivalent control strategy in \eqref{eqn:control-policy-v2} results in uniform boundedness of the mean squared states of the closed-loop system. Our proof of Theorem \ref{thm:msb} concludes by showing that the probability upper bound established in Lemma \ref{lemma:parameter-estimate-bound} decays sufficiently fast, satisfying the premise of Lemma \ref{lemma:msb-ce}.
\ssnote{\begin{remark}
    Larger $C$ is related to improved 'learnability' properties via a larger existent $p$ and $\Gamma_{sb}$ in Lemma \ref{lemma:satisfy-bmsb}, contributing to a larger $c_3$ in Lemma \ref{lemma:parameter-estimate-bound}, and hence faster exponential decay for the upper bound on $P( ||\hat{\theta}_i - \theta_*||_2 > d )$.
\end{remark}}
Only proof sketches that capture the main ideas are provided for the lemmas. Readers are referred to the Supplementary Materials for lengthy proofs of supporting results.

\subsection{Supporting Results}

\begin{lemma} \label{lemma:satisfy-bmsb}
Suppose A1-\ssnote{A2} hold on the closed-loop system \eqref{eqn:closed-loop-system-v2} \ssnote{and $x_0 \in \mathbb{R}$}. The following results hold on the sequence $(Z_t)_{t \in \mathbb{N}_0}$ from \eqref{eqn:state-input-data}:
\begin{enumerate}
    \item There exist $p>0$ and $\Gamma_{\text{sb}} \succ 0$ such that $(Z_t)_{t \in \mathbb{N}}$ satisfies the $(1,\Gamma_{\text{sb}},p)$-BMSB condition;
    \item $P\big( \sum_{t=1}^i Z_t Z_t^{\top} \not \preceq  \ssnote{\frac{1}{\delta} i((i ( \abs{b}U_{\text{max}} + \Sigma_W ) + |x_0|)^2} $ $ +U_{\text{max}}^2 ) I \big) \leq \delta$ holds for all $i \geq 1$ and $\delta \in (0,1)$.
\end{enumerate}
\end{lemma}
\begin{proof}[Proof Sketch] 
    For 1), let \ssnote{$(\mathcal{F}_t)_{t \in \mathbb{N}_0}$} be the natural filtration of \ssnote{$(Z_t)_{t \in \mathbb{N}_0}$}. Let $\gamma > 0$ satisfy $\mathbb{E}\left[ \abs{\zeta^{\top}Z_{t+1}} \mid \mathcal{F}_t \right] \geq \gamma$ for all $\zeta \in \mathcal{S}^1$ and $t \in \mathbb{N}_0$, where the existence of satisfactory values is established in Lemma \ref{lemma:paley-zygmund-satisfy-lower-bound} with the aid of a computer algebra system (CAS).
    For all $\zeta = (\zeta_1,\zeta_2) \in \mathcal{S}^1$ and $t \in \mathbb{N}_0$, $\Var{\zeta^{\top}Z_{t+1} \mid \mathcal{F}_t} \leq 2(\Sigma_W^2 + U_{\text{max}}^2)$ holds. Following an improvement of the Paley-Zygmund Inequality via the Cauchy-Schwarz inequality and making use of Jensen's inequality, 
    \begin{align}
        &P\left(\abs{\zeta^{\top}Z_{t+1}} > \sqrt{\zeta^{\top} \left( \psi^2 \gamma^2 I \right) \zeta} \mid \mathcal{F}_t \right)  \\
        & \quad \geq \bigg( 1 + \frac{\Var{\zeta^{\top} Z_{t+1} \mid \mathcal{F}_t}}{(1-\psi)^2 \mathbb{E}\left[ \abs{\zeta^{\top}Z_{t+1}} \mid \mathcal{F}_t \right]^2} \bigg)^{-1}
    \end{align}
    holds for all $\zeta = (\zeta_1,\zeta_2) \in \mathcal{S}^1$ and $\psi \in (0,1)$. Since $\Big( 1 + \frac{\Var{\zeta^{\top} Z_{t+1} \mid \mathcal{F}_t}}{(1-\psi)^2 \mathbb{E}\left[ \abs{\zeta^{\top}Z_{t+1}} \mid \mathcal{F}_t \right]^2} \Big)^{-1} \geq \Big( 1 + \frac{2(\Sigma_W^2 + U_{\text{max}}^2)}{(1-\psi)^2 \gamma^2} \Big)^{-1}$ holds, result 1) follows by fixing $\psi \in (0,1)$, and setting $\Gamma_{\text{sb}} = \psi^2 \gamma^2 I$ and $p = \Big( 1 + \frac{2(\Sigma_W^2 + U_{\text{max}}^2)}{(1-\psi)^2 \gamma^2} \Big)^{-1}$.

    For 2), suppose $i \in \mathbb{N}$. Using A1-\ssnote{A2}, the summed trace of the expected covariates can be bounded as $\sum_{t=1}^i \trace{\mathbb{E}\left[ Z_t Z_t^{\top} \right]} \leq \ssnote{i((i ( \abs{b}U_{\text{max}} + \Sigma_W ) + |x_0|)^2+U_{\text{max}}^2 )}$. Next, supposing $\delta \in (0,1)$, Markov's inequality is used to derive the upper bound \ssnote{$P\big( \sum_{t=1}^TZ_t Z_t^{\top} \not \preceq \frac{1}{\delta}i((i ( \abs{b}U_{\text{max}} + \Sigma_W ) + |x_0|)^2+U_{\text{max}}^2 )I  \big) \leq \delta (i((i ( \abs{b}U_{\text{max}} + \Sigma_W ) + |x_0|)^2+U_{\text{max}}^2 )) ^{-1} \sum_{t=1}^i \trace{\mathbb{E}\left[ Z_t Z_t^{\top} \right]}$}. The conclusion follows after combining these results.
\end{proof}

\begin{proposition} \label{prop:estim-bound}
\cite[Theorem 2.4]{simchowitz2018learning}  Fix $\delta \in (0,1)$, $i \in \mathbb{N}$ and $0 \prec \Gamma_{\text{sb}} \preceq \overline{\Gamma}$. Suppose $(Z_t,Y_t)_{t =1}^i \in ( \mathbb{R}^2 \times \mathbb{R} )^i$ is a random sequence such that (a) $Y_t = \theta_*^{\top} Z_t + \eta_t$ for $t \leq i$, where $\eta_t \mid \mathcal{F}_{t-1}$ is mean-zero and $\Sigma^2$-sub-Gaussian with $\mathcal{F}_t$ denoting the sigma-algebra generated by $\eta_0, \hdots, \eta_t, Z_1, \hdots, Z_t$, (b) $Z_1,\hdots,Z_i$ satisfies the $(k, \Gamma_{\text{sb}}, p)$-BMSB condition, and (c) $P(\sum_{t=1}^i Z_t Z_t^{\top} \not \preceq i \overline{\Gamma}) \leq \delta$ holds. Then if
\begin{equation}
    i \geq \frac{10 k}{p^2}\Big(\log\Big(\frac{1}{\delta}\Big)+4\log(10/p)+\log \det (\overline{\Gamma}\Gamma_{\text{sb}}^{-1})\Big), \label{eqn:estim-bound-i-condition}
\end{equation}
we have
\begin{align}
    &P\Bigg(\norm{\hat{\theta}_i-\theta_*}_{2} > \frac{90 \Sigma}{p} \\
    & \times \sqrt{\frac{ 1 + 2 \log \frac{10}{p} + \log \det \overline{\Gamma} \Gamma_{\text{sb}}^{-1} + \log \big(\frac{1}{\delta}\big) }{i \lambda_{\text{min}}(\Gamma_{\text{sb}}) }}\Bigg) \leq 3 \delta, \label{eqn:estim-bound-conclusion}
\end{align}
where $\hat{\theta}_i \ssnote{= (\bm{Z}_i^{\top} \bm{Z}_i)^{\dagger}\bm{Z}_i^{\top} \bm{Y}_i \in} \arg \min_{\theta \in \mathbb{R}^{2}} \sum_{t=1}^i ||Y_t - \theta^{\top} Z_t||_2^2$, \ssnote{and $\bm{Z}_i=[Z_1,\hdots,Z_i]^{\top}$, $\bm{Y}_i = [Y_1,\hdots,Y_i]^{\top} $}.
\end{proposition}

\begin{lemma} \label{lemma:parameter-estimate-bound}
    Consider $\Gamma_{\text{sb}} \succ 0$, $p > 0 $ and $q > 0$, $(Z_t)_{t \in \mathbb{N}}$ from \eqref{eqn:state-input-data}, and $(\hat{\theta}_t)_{t \in \mathbb{N}}$ from \eqref{eqn:parameter-estimate}. Suppose A1 holds on the closed-loop system \eqref{eqn:closed-loop-system-v2}. If the following are true:
    \begin{enumerate}
        \item $(Z_t)_{t \in \mathbb{N}}$ satisfies the $(1,\Gamma_{\text{sb}},p)$-BMSB condition;
        \item $P\big( \sum_{t=1}^i Z_t Z_t^{\top} \not \preceq \frac{1}{\delta} i^3 q I \big) \leq \delta$ holds for all $i \geq 1$ and $\delta \in (0,1)$;
    \end{enumerate}
    then for all $d \in \Big(0,\frac{ 90 \Sigma_W }{\sqrt{10 \lambda_{\text{min}} \left( \Gamma_{\text{sb}} \right)}}\Big)$ and $i \geq M(d,p,q,\Gamma_{\text{sb}})$,
    \begin{align}
        &P\left( \norm{\hat{\theta}_i - \theta_*}_2 > d \right) \leq i^{\frac{4}{3}} e^{- c_3(d,p,\Gamma_{\text{sb}}) i } c_4(p,q,\Gamma_{\text{sb}}).
    \end{align}
    Here, $c_1$, $c_2$, $c_3$, $c_4$ are defined as
    \begin{align}
        &c_1(q,\Gamma_{\text{sb}}) := q + \lambda_{\text{max}}(\Gamma_{\text{sb}}), \ c_2(p):=1 + 2 \log \left( 10 / p \right), \\
        &c_3(d,p,\Gamma_{\text{sb}}) := \frac{\lambda_{\text{min}} \left( \Gamma_{\text{sb}} \right) d^2 p^2}{3 ( 90 \Sigma_W )^2},  \\
        &c_4(p,q,\Gamma_{\text{sb}}):=3c_1(q,\Gamma_{\text{sb}})^{\frac{2}{3}} e^{\frac{1}{3}\left( c_2(p) + \log\left( \det \left( \Gamma_{\text{sb}}^{-1} \right) \right) \right)}
    \end{align}
    and $M(d,p,q,\Gamma_{\text{sb}})$, $M'(p,q,\Gamma_{\text{sb}})$ are defined as
    \begin{equation} \label{eqn:M-definition}
        \begin{aligned}
        &M(d,p,q,\Gamma_{\text{sb}}) := \max \bigg( \bigg \lceil \bigg( \frac{p^2}{10} - \frac{\lambda_{\text{min}} \left( \Gamma_{\text{sb}} \right) d^2 p^2}{ ( 90 \Sigma_W )^2} \bigg)^{-1} \\
        & \times \quad \left(4 \log \left( 10 / p \right) - c_2(p) \right) \bigg \rceil, M'(d,p,q,\Gamma_{\text{sb}})\bigg), \\
        &M'(d,p,q,\Gamma_{\text{sb}}) := \min \Big \{ m \in \mathbb{N} \mid \big( \forall i \geq m \big) \\
        & \quad \Big[ \frac{1}{3}i^{\frac{4}{3}} e^{- c_3(d,p,\Gamma_{\text{sb}}) i } c_4(p,q,\Gamma_{\text{sb}}) < 1 \Big] \Big\}.
        \end{aligned}
    \end{equation}
    for $d, p, q > 0$ and $\Gamma_{\text{sb}} \succ 0$.
\end{lemma}

\begin{proof}[Proof Sketch]
    Suppose $d \in (0,\frac{ 90 \Sigma_W }{\sqrt{10\lambda_{\text{min}} \left( \Gamma_{\text{sb}} \right)}})$, and $i \geq M(d,p,q,\Gamma_{\text{sb}})$. Let 
    \begin{align}
        \delta &= \frac{1}{3}i^{\frac{4}{3}} e^{- c_3(d,p,\Gamma_{\text{sb}}) i } c_4(p,q,\Gamma_{\text{sb}}), \\
        \overline{\Gamma}&=(1/\delta)i^2 c_1(q,\Gamma_{\text{sb}}) I .
    \end{align}
    The proof proceeds by showing that $(Z_t,\ssnote{X_{t+1}})_{t\in\mathbb{N}}$ from \eqref{eqn:state-input-data} satisfies the premise of Proposition \ref{prop:estim-bound}. From our selection of $d, i, \delta$ and $\overline{\Gamma}$, $\delta \in (0,1)$ and $\Gamma_{\text{sb}} \preceq \overline{\Gamma}$ both hold. Next, we prove conditions (a)-(c) in Proposition \ref{prop:estim-bound}. In particular, (a) \ssnote{$X_{t+1} = \theta_*^{\top} Z_t + W_t$} for $t \in \mathbb{N}$ holds with $\ssnote{W_t} \mid \mathcal{F}_{t-1}$ mean-zero and $\Sigma_W^2$-sub-Gaussian due to A1 when $\mathcal{F}_t$ is the sigma-algebra generated by $\ssnote{W_0},\hdots,W_t,Z_1,\hdots,Z_t$ for $t \in \ssnote{\mathbb{N}_0}$, (b) $(Z_1,\hdots,Z_i)$ satisfies the $(1,\Gamma_{\text{sb}},p)$-BMSB condition by 1) in our premise, and (c) $P\big( \sum_{t=1}^i Z_t Z_t^{\top} \not \preceq \frac{1}{\delta} i^3 q I \big) \leq \delta$ is satisfied by 2) in our premise. Moreover, \eqref{eqn:estim-bound-i-condition} is satisfied since $i \geq M(d,p,q,\Gamma_{\text{sb}})$ implies $i \geq (10/p^2) ( \log(\frac{1}{\delta})+(4 \log \left( 10 / p \right)) + \log( \det ( \overline{\Gamma}\Gamma_{\text{sb}}^{-1})))$, thereby allowing us to establish from \eqref{eqn:estim-bound-conclusion} that
    \begin{align}
        &P\Bigg(\norm{\hat{\theta}_i - \theta_*}_2 > ( 90 \Sigma_W / p ) \\
        &\times \sqrt{\frac{c_2(p) + \log ( \det ( \overline{\Gamma} \Gamma_{\text{sb}}^{-1} ) ) + \log ( \frac{1}{\delta})}{i \lambda_{\text{min}} \left( \Gamma_{\text{sb}} \right)}} \Bigg) \leq 3 \delta \\
        &= i^{\frac{4}{3}} e^{- c_3(d,p,\Gamma_{\text{sb}}) i } c_4(p,q,\Gamma_{\text{sb}}).
    \end{align}
    The conclusion follows by noting that $d \geq ( 90 \Sigma_W / p ) \sqrt{\frac{c_2(p) + \log ( \det ( \overline{\Gamma} \Gamma_{\text{sb}}^{-1} ) ) + \log ( \frac{1}{\delta})}{i \lambda_{\text{min}} \left( \Gamma_{\text{sb}} \right)}} $ holds.
\end{proof}

\begin{lemma} \label{lemma:msb-ce}
    Consider $(\hat{\theta}_t)_{t \in \mathbb{N}}$ from \eqref{eqn:parameter-estimate} and $(X_t)_{t \in \mathbb{N}_0}$ from \eqref{eqn:closed-loop-system-v2}. Suppose A1-\ssnote{A2} hold on the closed-loop system \eqref{eqn:closed-loop-system-v2} \ssnote{and $x_0 \in \mathbb{R}$}. If there exists $m > 0$ such that $\sum_{k=1}^{\infty}k^2 \sum_{i \geq k-1} P\left(  \norm{\hat{\theta}_i - \theta_*}_2 >d \right) < \infty$ for all $d \in (0,m)$, then
    \begin{equation}
        \exists e > 0: \sup_{t \in \mathbb{N}_0}\mathbb{E}\left[ X_t^2 \right] < e.
    \end{equation}
\end{lemma}

\begin{proof}[Proof Sketch]
    The proof proceeds by considering Case 1) $a \in (-1,1)$, and Case 2) $a \in \{-1,1\}$.

    \textit{Case 1:} Suppose $a \in (-1,1)$. From Lemma \ref{lemma:stability-a-not-1}, the process $(X_t)_{t \in \mathbb{N}_{0}}$ satisfies $\mathbb{E}\left[ X_t^2 \right] \leq \ssnote{x_0^2} + \frac{\beta(\lambda)}{1 - \lambda}$ for all $\lambda \in (a^2,1)$, with $\beta(\lambda)$ defined in Lemma \ref{lemma:stability-a-not-1}. The conclusion follows by choosing $\lambda \in (a^2,1)$, and setting $e = \ssnote{x_0^2} + \frac{\beta(\lambda)}{1-\lambda}$.
    
    \textit{Case 2:} Suppose $a \in \{-1,1\}$. Let $(X_t^*)_{t \in \mathbb{N}_0}$ taking values in $\mathbb{R}$ be the states of a \textit{reference system} satisfying
    \begin{equation}
        X_{t+1}^*=aX_t^*+b(\sigma_D(-(a/b)X_t^*)+V_t)+W_t, \label{eqn:reference-system}
    \end{equation}
    where $X_0^* = x_0$. The upper bound $\mathbb{E}\left[ X_t^2 \right] \leq 2(\mathbb{E}\left[ (X_t^*)^2 \right] + \mathbb{E}\left[ (X_t-X_t^*)^2 \right])$ holds for $t \in \mathbb{N}_0$. 
    
    The importance of the reference system is twofold.
    Firstly, its control strategy is not adaptive, allowing us to establish mean square boundedness of $(X_t^*)_{t\in \mathbb{N}_0}$ by using previously existing analysis techniques for when the system parameters are known, akin to \cite{chatterjee2012mean}. This is proven in Lemma \ref{lemma:stability-reference-system} by deriving the upper bound $\mathbb{E}\left[(X_t^*)^2\right] \leq \mathbb{E}\left[ ((a^t X_t^*)^+)^2 \right] + \mathbb{E}\left[ ((-a^t X_t^*)^+)^2 \right]$, as well as the fact that the the auxiliary sequences $(a^t X_t^*)_{t \in \mathbb{N}_0}$ and $(-a^t X_t^*)_{t \in \mathbb{N}_0}$ \ssnote{(where $a^t$ denotes `$a$ raised to $t$')} satisfy the conditions for Proposition \ref{prop:constant-drift-conditions} --- a result from \cite{pemantle1999moment} which provides moment conditions to uniformly bound a sequence with negative drift --- allowing us to uniformly bound $\mathbb{E}\left[((a^t X_t^*)^+)^2\right]$ and $\mathbb{E}\left[((-a^t X_t^*)^+)^2\right]$ from above.
        
    Secondly, the uniform boundedness of $\mathbb{E}\left[ (X_t-X_t^*)^2 \right]$ is derived by analyzing $(X_t-X_t^*)^2$. Define the  time that $\hat{\theta}_t$ enters and remains in a ball of size $d>0$ around $
    \theta_*$ as
    \begin{equation}
        T_{d} := \inf \left\{ t \in \mathbb{N} \mid  \norm{\hat{\theta}_i - \theta_*}_2 \leq d \text{ for all } i \geq t \right\}. \label{eqn:definition-Td}
    \end{equation}
    Moreover, define $d^*:=\min(m/2,1/2,\abs{b}/2)$. From Lemma \ref{lemma:bound-reference-error-after-enter-tube}, we establish that
    $
        (X_t - X_t^*)^2 \leq \max \Big( 4b^2D^2 \ssnote{(k+1)^2}, 4 \Big( \frac{\abs{b}+d^*}{1-d^*} + 3 \abs{b} \Big)^2D^2 \Big)
    $
    on the event $\{T_{d^*} = k \}$, which by the monotonicity of conditional expectation implies $\mathbb{E}\left[ (X_t - X_t^*)^2 \mid  T_{d^*}=k \right] \leq \max \big( 4b^2D^2 \ssnote{(k+1)^2}, 4\big( \frac{\abs{b}+d^*}{1-d^*} + 3 \abs{b} \big)^2D^2 \big)$ for $k \in \mathbb{N}$. Using this upper bound, the law of total expectation, and the fact that $P(T_{d^*} \geq k) \leq \sum_{i \geq k-1} P\Big(\norm{\hat{\theta}_i - \theta_*}_2 > d^*\Big)$ for $k \in \mathbb{N}$, the upper bound
    \begin{align}
        &\mathbb{E}\left[ (X_t-X_t^*)^2 \right] \leq \sum_{k = 1}^{\infty} \max \Big( 4b^2D^2 \ssnote{(k + 1)^2}, 4\Big( \frac{\abs{b}+d^*}{1-d^*}  \\
        & \quad + 3 \abs{b} \Big)^2D^2 \Big) \sum_{i \geq k-1} P\Big(\norm{\hat{\theta}_i - \theta_*}_2 > d^*\Big) \label{eqn:msb-proof-sketch-ineq-1} 
    \end{align}
    holds for $t \in \mathbb{N}_0$. The assumption that  $\sum_{k=1}^{\infty}k^2 \sum_{i \geq k-1} P\left( \norm{\hat{\theta}_i - \theta_*}_2 >d  \right) < \infty$ for all $d \in (0,m)$ from the premise implies that \eqref{eqn:msb-proof-sketch-ineq-1} is finite. Uniform boundedness of $\mathbb{E}\left[ X_t^2 \right]$ then follows.
\end{proof}

\subsection{Proof of Theorem \ref{thm:msb}}

\begin{proof}
    Let $p>0$ and $\Gamma_{\text{sb}} \succ 0$ be such that $(Z_t)_{t \in \mathbb{N}}$ satisfies the $(1,\Gamma_{\text{sb}},p)$-BMSB condition in Definition \ref{def:bmsb}, where the existence of satisfactory $p$ and $\Gamma_{\text{sb}}$ is established in 1) in Lemma \ref{lemma:satisfy-bmsb}. Moreover, let \ssnote{$q = (\abs{b}U_{\text{max}}+\Sigma_W + |x_0|)^2 + U_{\text{max}}^2$}. Since \ssnote{$(i ( \abs{b}U_{\text{max}} + \Sigma_W ) + |x_0|)^2+U_{\text{max}}^2  \leq i^2 q$} for $i \in \mathbb{N}$, making use of 2) in Lemma \ref{lemma:satisfy-bmsb} it follows that $P\big( \sum_{t=1}^TZ_t Z_t^{\top} \not \preceq \frac{1}{\delta}i^3 q I \big) \leq \delta$ holds for all $i \in \mathbb{N}$ and $\delta \in (0,1)$. With this, we have established that the premise of Lemma \ref{lemma:parameter-estimate-bound} is satisfied, and so we find
    \begin{align}
        &P\left( \norm{\hat{\theta}_i - \theta_*}_2 > d  \right) \leq i^{\frac{4}{3}} e^{- c_3(d,p,\Gamma_{\text{sb}}) i } c_4(p,q,\Gamma_{\text{sb}}) \label{eqn:msb-ineq-4}
    \end{align}
    holds for all $d \in \Big(0,\frac{ 90 \Sigma_W }{\sqrt{10 \lambda_{\text{min}} \left( \Gamma_{\text{sb}} \right)}}\Big)$ and $i \geq M(d,p,q,\Gamma_{\text{sb}})$.

    Now, suppose $d \in \Big(0,\frac{ 90 \Sigma_W }{\sqrt{10 \lambda_{\text{min}} \left( \Gamma_{\text{sb}} \right)}}\Big)$. For ease of readability, let us refer to the functions $c_1,c_2,c_3,c_4$ and $M$ without their arguments, but with an implicit understanding of their dependence on $d,p,q,\Gamma_{\text{sb}}$. The following holds:
    \begin{align}
        &\sum_{k=1}^{\infty}k^2 \sum_{i \geq k-1} P\left(  \norm{\hat{\theta}_i - \theta_*}_2 >d  \right) \\
        & \leq \sum_{k=1}^{\infty}k^2 \sum_{i=k-1}^{M-1} 1 +c_4 \sum_{k=1}^{\infty}k^2 \sum_{i=\max(k-1,M)}^{\infty} i^{\frac{4}{3}} e^{- c_3 i } \label{eqn:msb-ineq-3} \\
        &< \infty \label{eqn:msb-ineq-1} 
    \end{align}
    where \eqref{eqn:msb-ineq-3} follows from \eqref{eqn:msb-ineq-4}, and \eqref{eqn:msb-ineq-1} follows from $\sum_{k=1}^{\infty}k^2 \sum_{i=k-1}^{M-1} 1 =  \sum_{k=1}^{M} k^2 \left({M}-k + 1 \right) < \infty$, as well as
    \begin{align}
         &\sum_{k=1}^{\infty}k^2 \sum_{i=\max(k-1,M)}^{\infty} i^2 e^{- c_3 i } \\
         &= \sum_{k=1}^{M}k^2 \sum_{i=M}^{\infty} i^2 e^{- c_3 i } + \sum_{k=M+1}^{\infty}k^2 \sum_{i=k-1}^{\infty} i^2 e^{- c_3 i } \\
         & < \infty. \label{eqn:msb-ineq-2}
    \end{align}
    Here, \eqref{eqn:msb-ineq-2} follows since $\sum_{k=1}^{M}k^2 \sum_{i=M}^{\infty} i^2 \expf{- c_3 i } = \sum_{k=1}^M k^2 \frac{e^{{c_3}-{c_3} M} \left(-2 e^{c_3} M^2+e^{2 {c_3}} M^2+2 e^{c_3} M+e^{c_3}+M^2-2 M+1\right)}{\left(e^{c_3}-1\right)^3}  < \infty$,
    and
    $
        \sum_{k=M+1}^{\infty}k^2 \sum_{i=k-1}^{\infty} i^2 \expf{- c_3 i } =\frac{e^{-{c_3} (M-2)}}{(e^{c_3}-1)^6} ( (M-2)^2 M^2 \\
         \quad + e^{4 {c_3}} (M+1)^2 M^2 -e^{3 {c_3}} (M+1)^2 (4 M^2-6 M-5) \\
         \quad + e^{2 {c_3}} (6 M^4-6 M^3-25 M^2+8 M+26) \\
         \quad +e^{c_3} (-4 M^4+10 M^3+7 M^2-24 M+9)) < \infty.
    $
    
    From \eqref{eqn:msb-ineq-1}, we have shown that $\sum_{k=1}^{\infty}k^2 \sum_{i \geq k-1} P\left( \norm{\hat{\theta}_i - \theta_*}_2 >d  \right) < \infty$ for all $d \in \Big(0,\frac{ 90 \Sigma_W }{\sqrt{10 \lambda_{\text{min}} \left( \Gamma_{\text{sb}} \right)}}\Big)$. The premise of Lemma \ref{lemma:msb-ce} is thus satisfied with $m=\frac{ 90 \Sigma_W }{\sqrt{10 \lambda_{\text{min}} \left( \Gamma_{\text{sb}} \right)}}$. The conclusion follows.
\end{proof}

\section{Numerical Examples}
To demonstrate the effectiveness of the control strategy in Algorithm 1, we tested it with $U_{\text{max}} = 1$, $C = 0.1$ and $\hat{\theta}_{\text{init}} = (-1,-5)$ on three different systems with $x_0=0$:
\begin{itemize}
    \item System 1: $a = 0.7$, $b = -1$, $\Sigma_W = 1$;
    \item System 2: $a = -1$, $b = 2$, $\Sigma_W = 2$;
    \item System 3: $a = 1$, $b = 0.5$, $\Sigma_W = 1.5$.
\end{itemize}
Fig. \ref{fig:sys-1-sim} shows the empirical ensemble average of $X_t^2$ for Systems 1, 2 and 3 respectively over 1000 runs. Our control strategy seemingly attains mean square boundedness for all three systems, matching the guarantee provided in Theorem \ref{thm:msb}. We simulate System 3 with no controls for comparison, which does not achieve mean square boundedness. \ssnote{Convergence of $\hat{a}_t$ and $\hat{b}_t$ over time are shown in Fig. \ref{fig:sys-1to3-ahat} and \ref{fig:sys-1to3-bhat}.}

\begin{figure}
    \centering
    \begin{subfigure}[b]{0.49\textwidth}
        \centering
        \includegraphics[width=0.6\textwidth]{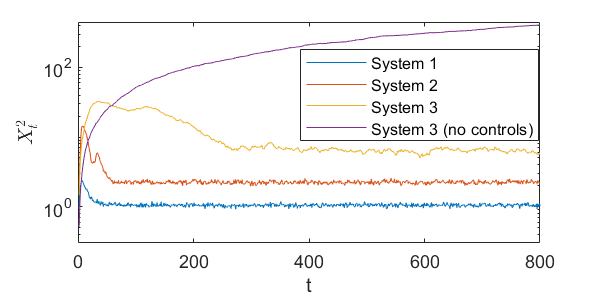}
        \caption{Log scale plot of $X_t^2$}
        \label{fig:sys-1-sim}
    \end{subfigure}
    \begin{subfigure}[b]{0.24\textwidth}
        \centering
        \includegraphics[width=1\textwidth]{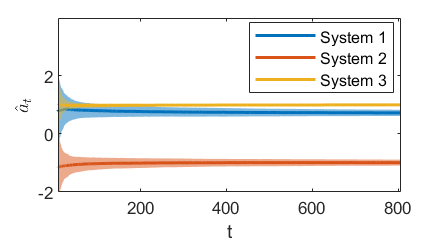}
        \caption{\ssnote{Mean/standard devation of $\hat{a}_t$}}
        \label{fig:sys-1to3-ahat}
    \end{subfigure}%
    \hfill
    \begin{subfigure}[b]{0.24\textwidth}
        \centering
        \includegraphics[width=1\textwidth]{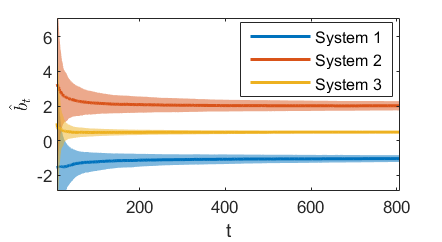}
        \caption{\ssnote{Mean/standard devation of $\hat{b}_t$}}
        \label{fig:sys-1to3-bhat}
    \end{subfigure}
    \caption{\ssnote{Ensemble average plots for Systems 1, 2, 3, and 3 (no controls) over $1000$ runs.}}
        \label{fig:sys-1to3-thetahat}
\end{figure}

\section{Conclusion}
We proposed a perturbed CE control scheme for adaptive control of stochastic, scalar, at-worst marginally stable linear systems subject to additive, i.i.d. Gaussian disturbances, with positive upper bound constraints on the control magnitude. Mean square boundedness of the closed-loop system is established, and demonstrated by numerical examples.

\ssnote{It is possible to consider non-Gaussian stochastic processes in A1, and establish mean square boundedness. The most critical requirements are ensuring $W_{t} \mid \mathcal{F}_{t-1}$ is mean-zero and $\Sigma^2$ sub-Gaussian, and proving Lemma \ref{lemma:paley-zygmund-satisfy-lower-bound}. The latter requires careful inspection of the particular disturbance distribution, the excitation term, and the nonlinear saturation.}

\ssnote{Our approach has a strong potential to be extended to higher dimensions. The core of our method is combining model-based control in Line 6 of Algorithm \ref{alg:ideal-controller-v2} with least squares parameter estimation in Line 9. Stability analysis follows by satisfying Lemma \ref{lemma:satisfy-bmsb} to establish fast convergence of upper bounds on $P(||\hat{\theta}_i-\theta_* ||_2 > d)$ in Lemma \ref{lemma:parameter-estimate-bound}, and proving that fast convergence implies mean square boundedness in Lemma \ref{lemma:msb-ce}. This intuition generalizes to higher dimensions, but to make the jump analytically, some technical challenges remain to be solved. In particular, the careful analysis of 1) in Lemma \ref{lemma:satisfy-bmsb} needs to be scaled up from the 1D case, and an equivalent result to Lemma \ref{lemma:msb-ce} is required, since Lemma \ref{lemma:bound-reference-error-after-enter-tube} for bounding $||X_t-X_t^*||_2$ does not immediately hold in $n$ dimensions.
}




\bibliography{references.bib}
\bibliographystyle{ieeetr.bst}

\appendix
\label{Appendix}

\begin{lemma} \label{lemma:paley-zygmund-satisfy-lower-bound}
Suppose A1-\ssnote{A2} hold on the closed-loop system \eqref{eqn:closed-loop-system-v2} \ssnote{and $x_0 \in \mathbb{R}$}. Let $\ssnote{(\mathcal{F}_t)_{t \in \mathbb{N}_0}}$ be the natural filtration of \ssnote{$(Z_t)_{t \in \mathbb{N}_0}$} from \eqref{eqn:state-input-data}. There exists $\gamma > 0$ such that for all $\zeta \in \mathcal{S}^1$ and $t \geq 0$, 
\begin{align}
    \mathbb{E}\left[ \abs{\zeta^{\top}Z_{t+1}} \mid \mathcal{F}_t \right] \geq \gamma
\end{align}
\end{lemma}

\begin{lemma} \label{lemma:stability-a-not-1}
    Consider the states $(X_t)_{t \in \mathbb{N}_0}$ from the closed-loop system \eqref{eqn:closed-loop-system-v2}. Suppose A1-\ssnote{A2} hold\ssnote{, $x_0 \in \mathbb{R}$}, and $a \in (-1,1)$. For all $\lambda \in (a^2,1)$ and $t \in \mathbb{N}_0$, we have
    \begin{equation}
        \mathbb{E}\left[ X_t^2 \right]  \leq \ssnote{x_0^2} + \frac{\beta(\lambda)}{1 - \lambda}.
    \end{equation}
    where
    \begin{align}
    \beta(\lambda) &:= a^2 E(\lambda)^2 + 2 \abs{a} E(\lambda) D_1 + D_2, \\
    E(\lambda) &:= \frac{\abs{a}D_1 + \sqrt{a^2D_1^2 + (\lambda - a^2)D_2}}{\lambda - a^2}, \\
    D_1 &:= \abs{b}U_{\text{max}} + S_1 , \\
    D_2 &:= b^2 U_{\text{max}}^2 + 2 \abs{b} U_{\text{max}} S_1 + S_2.
\end{align}
\end{lemma}

\begin{lemma} \label{lemma:stability-reference-system}
Consider the states $(X_t^*)_{t \in \mathbb{N}_0}$ from the reference system \eqref{eqn:reference-system}. Suppose A1-\ssnote{A2} hold\ssnote{, $x_0 \in \mathbb{R}$,} and $a \in \{-1,1\}$. There exists $e>0$ such that for all $t \in \mathbb{N}_0$, $\mathbb{E}\left[ \left(X_t^*\right)^2 \right] \leq e$.
\end{lemma}

\begin{proposition} \label{prop:constant-drift-conditions} \cite[Theorem 1]{pemantle1999moment}
Let $(\xi_t)_{t\in \mathbb{N}_0}$ be a sequence of scalar random variables 
and let $(\mathcal{F}_t)_{t \in \mathbb{N}_0}$ be any filtration to which $(\xi_t)_{t \in \mathbb{N}_0}$ is adapted. Suppose that there exist constants $\gamma > 0$ and $J, M < \infty$, such that $\xi_0 \leq J$, and for all $t$:
\begin{equation}
    \mathbb{E}\left[ \xi_{t+1} - \xi_t | \mathcal{F}_t \right] \leq -\gamma \text{ on the event } \{\xi_t > J \}, \label{eqn:msb-condition-1}
\end{equation}
and
\begin{equation}
    \mathbb{E}\left[ |\xi_{t+1} - \xi_t|^4 \Big| \xi_0, \hdots, \xi_t \right] \leq M. \label{eqn:msb-condition-2}
\end{equation}
Then there exists a constant $c > 0$ such that
\begin{equation}
    \sup_{t \in \mathbb{N}_0} \mathbb{E}\left[ (\xi_t^+)^2 \right] \leq c. \label{eqn:msb-conclusion}
\end{equation}
\end{proposition}

\begin{lemma} \label{lemma:bound-reference-error-after-enter-tube}
Consider the states $(X_t)_{t \in \mathbb{N}_0}$ from the closed-loop system \eqref{eqn:closed-loop-system-v2}, and $(X_t^*)_{t \in \mathbb{N}_0}$ from the reference system \eqref{eqn:reference-system}. Suppose A1-\ssnote{A2} hold, \ssnote{$x_0 \in \mathbb{R}$, }and $a \in \{-1,1\}$, for both the closed-loop system and reference system. For all $d \in (0,\min(1,\abs{b}))$ , $k \in \mathbb{N}$ and $t \in \mathbb{N}_0$, on the event $\{ T_{d} = k \}$, $\abs{X_t - X_t^*} \leq \max \left( \ssnote{(k + 1)} \abs{b} 2D, 2 \left(\frac{\abs{b}+d}{1-d} + 3\abs{b} \right) D  \right)$ (with $T_d$ defined in \eqref{eqn:definition-Td}).
\end{lemma}

\section*{Supplementary Materials}
A dependency graph for the theoretical results in this work is illustrated in Figure \ref{fig:results-dependency-graph}.

\begin{figure}[thpb]
  \centering
  \includegraphics[width=0.49\textwidth]{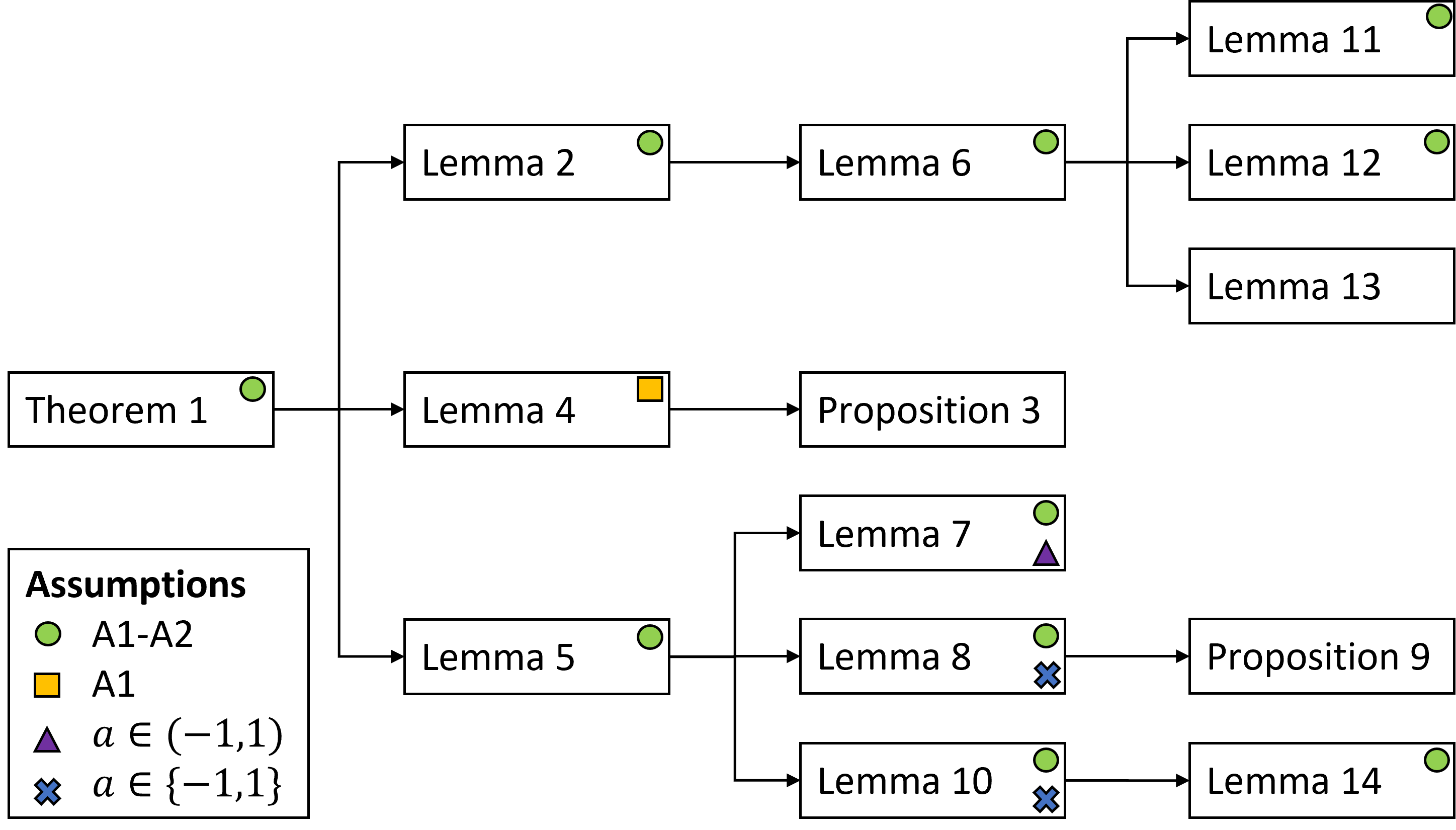}
  \caption{Dependency graph for theoretical results in this work.}
  \label{fig:results-dependency-graph}
\end{figure}

\subsection{Analysis for Lemma \ref{lemma:satisfy-bmsb}}
We start by providing the proof of Lemmas  \ref{lemma:satisfy-bmsb} and \ref{lemma:paley-zygmund-satisfy-lower-bound}. Following this, we provide Lemmas \ref{lemma:lower-bound-A}, \ref{lemma:lower-bound-B} and \ref{lemma:properties-f} and their proofs, which are supporting results used to prove Lemma \ref{lemma:paley-zygmund-satisfy-lower-bound}.
\begin{proof}[Proof of Lemma \ref{lemma:satisfy-bmsb}]
    We start by proving 1). Let $(\mathcal{F}_t)_{t \in \mathbb{N}_0}$ be the natural filtration of $(Z_t)_{t \in \mathbb{N}_0}$. Let $\gamma > 0$ satisfy $\mathbb{E}\left[ \abs{\zeta^{\top}Z_{t+1}} \mid \mathcal{F}_t \right] \geq \gamma$ for all $\zeta \in \mathcal{S}^1$ and $t \in \mathbb{N}_0$, where the existence of satisfactory values is established in Lemma \ref{lemma:paley-zygmund-satisfy-lower-bound}. Now, suppose $\zeta = (\zeta_1,\zeta_2) \in \mathcal{S}^1$, and $t \in \mathbb{N}_0$. We have
    \begin{align}
    &\Var{\zeta^{\top}Z_{t+1} \mid \mathcal{F}_t}\\ 
    &= \mathbb{E}[ (( \zeta_1 X_{t+1} + \zeta_2 U_{t+1}) - \mathbb{E}[ \zeta_1 X_{t+1} + \zeta_2 U_{t+1} \mid \mathcal{F}_t ])^2 \mid \mathcal{F}_t ] \\
    & \leq 2(\Var{X_{t+1} \mid \mathcal{F}_t} + \Var{U_{t+1} \mid \mathcal{F}_t} ) \label{eqn:paley-zygmund-satisfy-upper-bound-ineq-2}
    \end{align}
    where \eqref{eqn:paley-zygmund-satisfy-upper-bound-ineq-2} follows from $(a+b)^2 \leq 2(a^2 + b^2)$ for $a,b \in \mathbb{R}$, linearity of expectation, and $\zeta_1^2,\zeta_2^2 \leq 1$. Focusing on $\Var{X_{t+1} \mid \mathcal{F}_t}$, we have the equality
    \begin{align}
        \Var{X_{t+1} \mid \mathcal{F}_t} &= \Var{\theta_*^{\top} Z_t + W_t \mid \mathcal{F}_t} \\
        &= \Var{\theta_*^{\top} Z_t + W_t \mid Z_t} \\
        &= \Var{W_t \mid Z_t} \label{eqn:paley-zygmund-satisfy-upper-bound-ineq-3} \\
        &= \Sigma_W^2 \label{eqn:paley-zygmund-satisfy-upper-bound-ineq-4}.
    \end{align}
    where \eqref{eqn:paley-zygmund-satisfy-upper-bound-ineq-3} holds since $Z_t$ is $\mathcal{F}_t$-measurable.
    Since $U_{t+1}$ only takes on values in $[-U_{\text{max}}, U_{\text{max}}]$, we have via Popoviciu's inequality,
    \begin{align}
        \Var{U_{t+1} \mid \mathcal{F}_t} \leq \frac{1}{4}\left( U_{\text{max}} - (-U_{\text{max}}) \right)^2 = U_{\text{max}}^2. \label{eqn:paley-zygmund-satisfy-upper-bound-ineq-5}
    \end{align}
    Combining \eqref{eqn:paley-zygmund-satisfy-upper-bound-ineq-2}, \eqref{eqn:paley-zygmund-satisfy-upper-bound-ineq-4}, and \eqref{eqn:paley-zygmund-satisfy-upper-bound-ineq-5}, we derive the following upper bound:
    \begin{equation}
        \Var{\zeta^{\top}Z_{t+1} \mid \mathcal{F}_t} \leq 2(\Sigma_W^2 + U_{\text{max}}^2). \label{eqn:paley-zygmund-satisfy-upper-bound-ineq}
    \end{equation}
    Next, note that the following inequality holds for all $\psi \in (0,1)$:
\begin{align}
    &P\left(\abs{\zeta^{\top}Z_{t+1}} > \sqrt{\zeta^{\top} \left( \psi^2 \gamma^2 I \right) \zeta} \mid \mathcal{F}_t \right) \\
    &= P\left( \abs{\zeta^\top Z_{t+1}} > \psi \gamma \mid \mathcal{F}_t \right) \label{eqn:paley-zygmund-sufficient-conditions-ineq-5} \\
    &\geq P(\abs{\zeta^{\top}Z_{t+1}} > \psi \mathbb{E}\left[ \abs{\zeta^{\top}Z_{t+1}} \mid \mathcal{F}_t \right] \mid \mathcal{F}_t ) \label{eqn:paley-zygmund-sufficient-conditions-ineq-1} \\
    &\geq \left( 1 + \frac{\Var{\abs{\zeta^{\top} Z_{t+1}} \mid \mathcal{F}_t}}{(1-\psi)^2 \mathbb{E}\left[ \abs{\zeta^{\top}Z_{t+1}} \mid \mathcal{F}_t \right]^2} \right)^{-1} \label{eqn:paley-zygmund-sufficient-conditions-ineq-2} \\
    &\geq \left( 1 + \frac{\Var{\zeta^{\top} Z_{t+1} \mid \mathcal{F}_t}}{(1-\psi)^2 \mathbb{E}\left[ \abs{\zeta^{\top}Z_{t+1}} \mid \mathcal{F}_t \right]^2} \right)^{-1} \label{eqn:paley-zygmund-sufficient-conditions-ineq-3} \\
    &\geq \left( 1 + \frac{2(\Sigma_W^2 + U_{\text{max}}^2)}{(1-\psi)^2 \gamma^2} \right)^{-1} \label{eqn:paley-zygmund-sufficient-conditions-ineq-4}
\end{align}
where \eqref{eqn:paley-zygmund-sufficient-conditions-ineq-5} holds since $\zeta^\top \zeta = 1$, \eqref{eqn:paley-zygmund-sufficient-conditions-ineq-1} follows from $\mathbb{E}\left[ \abs{\zeta^{\top}Z_{t+1}} \mid \mathcal{F}_t \right] \geq \gamma$, \eqref{eqn:paley-zygmund-sufficient-conditions-ineq-2} follows from an improvement of the Paley-Zygmund inequality via the Cauchy-Schwarz inequality, \eqref{eqn:paley-zygmund-sufficient-conditions-ineq-3} holds since for any random variable $X$ taking values in $\mathbb{R}$, $\Var{\abs{X}} \leq \Var{X}$ is true making use of Jensen's inequality. Finally, \eqref{eqn:paley-zygmund-sufficient-conditions-ineq-4} follows from $\mathbb{E}\left[ \abs{\zeta^{\top}Z_{t+1}} \mid \mathcal{F}_t \right] \geq \gamma$ and \eqref{eqn:paley-zygmund-satisfy-upper-bound-ineq}.
Fixing $\psi \in (0,1)$, and setting $\Gamma_{\text{sb}} = \psi^2 \gamma^2 I$ and $p = (1 + 2(\Sigma_W^2 + U_{\text{max}}^2) / ((1-\psi)^2 \gamma^2))^{-1}$, result 1) then follows.

We now prove 2). We start this by establishing that $\sum_{t=1}^i \mathbb{E}\left[ U_{t}^2 \right] \leq i U_{\text{max}}^2$. Using the fact that for all $t \in \mathbb{N}$, $U_{t}^2 \leq U_{\text{max}}^2$, we have,
\begin{align}
    \sum_{t=1}^i \mathbb{E}\left[ U_{t}^2 \right] \leq \sum_{t=1}^i U_{\text{max}}^2 = i U_{\text{max}}^2. \label{eqn:cov-upper-bound-ineq-9}
\end{align}
Next, we prove that $\sum_{t=1}^i \mathbb{E}\left[ X_{t}^2 \right] \leq i(i ( \abs{b}U_{\text{max}} + \Sigma_W ) + |x_0|)^2$. For all $t \in \mathbb{N}_0$, we have,
\begin{align}
    &\mathbb{E}[ X_{t+1}^2 ] \\
    &= \mathbb{E}[ ( a X_t + b U_t + W_t )^2 ] \\
    &= a^2 \mathbb{E}[X_t^2] + 2 a b \mathbb{E}[X_t U_t] + b^2 \mathbb{E}[U_t^2] + \Sigma_W^2 \\
    & \leq \mathbb{E}[X_t^2] + 2 \abs{b} U_{\text{max}} \sqrt{\mathbb{E}[ X_t^2 ]} + b^2 U_{\text{max}}^2 \label{eqn:cov-upper-bound-ineq-2} \\
    &\quad + \Sigma_W^2 \\
    &= ( \sqrt{\mathbb{E}[ X_t^2 ]} + \abs{b} U_{\text{max}} )^2 + \Sigma_W^2. \label{eqn:cov-upper-bound-ineq-3}
\end{align}
where \eqref{eqn:cov-upper-bound-ineq-2} follows from A2 and the Cauchy-Schwarz inequality, and \eqref{eqn:cov-upper-bound-ineq-3} follows via quadratic factorization. Taking the square root of both sides, we then have,
\begin{align}
    \sqrt{\mathbb{E}\left[ X_{t+1}^2 \right]} &\leq \sqrt{\left( \sqrt{\mathbb{E}\left[ X_t^2 \right]} + \abs{b} U_{\text{max}} \right)^2 + \Sigma_W^2} \\
    & \leq \sqrt{\mathbb{E}\left[ X_t^2 \right]} + \abs{b} U_{\text{max}} + \Sigma_W. \label{eqn:cov-upper-bound-ineq-4}
\end{align}
By iteratively applying \eqref{eqn:cov-upper-bound-ineq-4} and noting that $\mathbb{E}\left[ X_0^2 \right] = x_0^2$, we have for $t \geq 1$,
\begin{align}
    \sqrt{\mathbb{E}\left[ X_t^2 \right]} &\leq \sum_{s=0}^{t-1}\left( \abs{b} U_{\text{max}} + \Sigma_W \right) + |x_0| \\
    &= t ( \abs{b} U_{\text{max}} + \Sigma_W ) + |x_0|.
\end{align}
Squaring both sides, we then have
\begin{align}
    \mathbb{E}\left[ X_t^2 \right] \leq (t ( \abs{b}U_{\text{max}} + \Sigma_W ) + |x_0|)^2.
\end{align}
Summing from $t=1$ to $i$, we have,
\begin{align}
    &\sum_{t=1}^i \mathbb{E}\left[ X_t^2 \right] \\
    &\leq \sum_{t=1}^i  (t \left( \abs{b}U_{\text{max}} + \Sigma_W \right) + |x_0|)^2 \\
    &\leq \sum_{t=1}^i  (i \left( \abs{b}U_{\text{max}} + \Sigma_W \right) + |x_0|)^2 \\
    &= i(i ( \abs{b}U_{\text{max}} + \Sigma_W ) + |x_0|)^2. \label{eqn:cov-upper-bound-ineq-10}
\end{align}
Next, note that the following holds:
\begin{align}
&\sum_{t=1}^i \trace{ \mathbb{E}\left[ Z_t Z_t^{\top} \right]} \\
&= \sum_{t=1}^i \mathbb{E}\left[ X_{t}^2 \right] + \sum_{t=1}^i \mathbb{E}\left[ U_{t}^2 \right] \\
&\leq i ((i ( \abs{b}U_{\text{max}} + \Sigma_W ) + |x_0|)^2+U_{\text{max}}^2 ) \label{eqn:cov-upper-bound-ineq-8}
\end{align}
where \eqref{eqn:cov-upper-bound-ineq-8} follows from \eqref{eqn:cov-upper-bound-ineq-9} and \eqref{eqn:cov-upper-bound-ineq-10}. Finally, fix $\delta \in (0,1)$. We find that
\begin{align}
    &P ( \sum_{t=1}^i Z_t Z_t^{\top} \not \preceq \frac{1}{\delta} i((i ( \abs{b}U_{\text{max}} + \Sigma_W ) + |x_0|)^2+U_{\text{max}}^2 ) I ) \allowdisplaybreaks \\
    &= P \Big( \lambda_{\text{max}}( (i((i ( \abs{b}U_{\text{max}} + \Sigma_W ) + |x_0|)^2+U_{\text{max}}^2 ))^{-1}  \allowdisplaybreaks \\
    &\quad \times \sum_{t=1}^i Z_t Z_t^{\top} ) \geq \frac{1}{\delta} \Big) \allowdisplaybreaks \label{eqn:cov-upper-bound-eq-1} \\
    & \leq \delta \mathbb{E}\Big[ \lambda_{\text{max}}( (i((i ( \abs{b}U_{\text{max}} + \Sigma_W ) + |x_0|)^2+U_{\text{max}}^2 ))^{-1}  \allowdisplaybreaks \\
    &\quad \times \sum_{t=1}^i Z_t Z_t^{\top} )  \Big] \allowdisplaybreaks \label{eqn:cov-upper-bound-ineq-5} \\
    & \leq \delta \mathbb{E}\Big[ \text{tr} \Big( (i((i ( \abs{b}U_{\text{max}} + \Sigma_W ) + |x_0|)^2+U_{\text{max}}^2 ))^{-1}  \allowdisplaybreaks \\
    & \quad \times \sum_{t=1}^i Z_t Z_t^{\top} \Big)  \Big] \label{eqn:cov-upper-bound-ineq-6} \allowdisplaybreaks \\
    &= \delta (i((i ( \abs{b}U_{\text{max}} + \Sigma_W ) + |x_0|)^2+U_{\text{max}}^2 ))^{-1} \\
    &\quad \times \sum_{t=1}^i \trace{\mathbb{E}[Z_t Z_t^{\top} ] } \allowdisplaybreaks \\
    & \leq \delta \label{eqn:cov-upper-bound-ineq-7}
\end{align}
where \eqref{eqn:cov-upper-bound-eq-1} follows from the definition of $\preceq$, \eqref{eqn:cov-upper-bound-ineq-5} follows from Markov's inequality, \eqref{eqn:cov-upper-bound-ineq-6} holds since for a matrix $M \in \mathbb{R}^{d} \times \mathbb{R}^d$, $\lambda_{\text{max}}(M) \leq \trace{M}$, and \eqref{eqn:cov-upper-bound-ineq-7} follows from \eqref{eqn:cov-upper-bound-ineq-8}.
Thus, result 2) has been established.
\end{proof}

\begin{proof}[Proof of Lemma \ref{lemma:paley-zygmund-satisfy-lower-bound}]
For all $\zeta=(\zeta_1,\zeta_2) \in \mathcal{S}^1$ and $t \geq 0$, we have,
\begin{align}
    &\mathbb{E}[ \abs{\zeta^{\top} Z_{t+1}} \mid \mathcal{F}_t ] \\
    &= \mathbb{E}[ \abs{\zeta_1 X_{t+1} + \zeta_2 U_{t+1}} \mid \mathcal{F}_t ] \\
    &= \mathbb{E}[ \abs{\zeta_1 X_{t+1} + \zeta_2 ( \sigma_D (G_{t+1} X_{t+1}) + V_{t+1})} \mid \mathcal{F}_t ] \label{eqn:paley-zygmund-satisfy-lower-bound-eq-1} \\
    &= \mathbb{E}[ \abs{\zeta_1 X_{t+1} + \zeta_2 ( \sigma_D (G_{t+1} X_{t+1}) ) + \zeta_2 V_{t+1}} \mid \mathcal{F}_t ]. \label{eqn:paley-zygmund-satisfy-lower-bound-eq-2}
\end{align}
where \eqref{eqn:paley-zygmund-satisfy-lower-bound-eq-1} follows from \eqref{eqn:control-policy-v2}. A lower bound can be derived for \eqref{eqn:paley-zygmund-satisfy-lower-bound-eq-2} as follows:
\begin{align}
    &\mathbb{E}[ \abs{\zeta_1 X_{t+1} + \zeta_2 ( \sigma_D (G_{t+1} X_{t+1}) ) + \zeta_2 V_{t+1}} \mid \mathcal{F}_t ] \\
    &= \mathbb{E}[ \mathbb{E}[|\zeta_1 X_{t+1} + \zeta_2 ( \sigma_D (G_{t+1} X_{t+1}) ) \label{eqn:paley-zygmund-satisfy-lower-bound-eq-3} \\
    & \quad + \zeta_2 V_{t+1}| \mid X_{t+1}, \mathcal{F}_t ] \mid \mathcal{F}_t ] \\
    &\geq \mathbb{E}[ | \mathbb{E}[\zeta_1 X_{t+1} + \zeta_2 ( \sigma_D (G_{t+1} X_{t+1}) ) \label{eqn:paley-zygmund-satisfy-lower-bound-ineq-1} \\
    &\quad + \zeta_2 V_{t+1} \mid X_{t+1}, \mathcal{F}_t  ] |  \mid \mathcal{F}_t ]\\
    &= \mathbb{E}[ | \mathbb{E}[\zeta_1 X_{t+1} \label{eqn:paley-zygmund-satisfy-lower-bound-eq-4} \\
    & \quad + \zeta_2 ( \sigma_D (G_{t+1} X_{t+1}) ) \mid X_{t+1}, \mathcal{F}_t  ] |  \mid \mathcal{F}_t ] \\
    &= \mathbb{E}[ \abs{ \zeta_1 X_{t+1} + \zeta_2 ( \sigma_D (G_{t+1} X_{t+1}) ) }  \mid \mathcal{F}_t ] \label{eqn:paley-zygmund-satisfy-lower-bound-eq-5} \\
    &= \mathbb{E}[ \abs{ A_{t+1}(\zeta) }  \mid \mathcal{F}_t ] \label{eqn:paley-zygmund-satisfy-lower-bound-eq-5b}
\end{align}
where \eqref{eqn:paley-zygmund-satisfy-lower-bound-eq-3} follows from the tower property, \eqref{eqn:paley-zygmund-satisfy-lower-bound-ineq-1} follows from Jensen's inequality and the monotonocity of conditional expectation, \eqref{eqn:paley-zygmund-satisfy-lower-bound-eq-4} follows from the independence of $V_{t+1}$ and $X_{t+1}, \mathcal{F}_t$, and \eqref{eqn:paley-zygmund-satisfy-lower-bound-eq-5} follows since $\zeta_1 X_{t+1}+ \zeta_2 ( \sigma_D (G_{t+1} X_{t+1}) )$ is $X_{t+1},\mathcal{F}_t$-measurable, and \eqref{eqn:paley-zygmund-satisfy-lower-bound-eq-5b} follows by defining
\begin{equation}
    A_{t+1}(\zeta):= \zeta_1 X_{t+1} + \zeta_2 \left( \sigma_D (G_{t+1} X_{t+1}) \right), \quad t \in \mathbb{N}_0. \label{eqn:paley-zygmund-satisfy-lower-bound-define-A}
\end{equation}
Similarly,
\begin{align}
    &\mathbb{E}[ |\zeta_1 X_{t+1} + \zeta_2 ( \sigma_D (G_{t+1} X_{t+1}) ) + \zeta_2 V_{t+1}| \mid \mathcal{F}_t ] \\
    &\geq \mathbb{E}[ | \mathbb{E}[\zeta_1 X_{t+1} + \zeta_2 ( \sigma_D (G_{t+1} X_{t+1}) ) \label{eqn:paley-zygmund-satisfy-lower-bound-ineq-2} \\
    & \quad + \zeta_2 V_{t+1} \mid V_{t+1}, \mathcal{F}_t  ] |  \mid \mathcal{F}_t ]\\
    &= \mathbb{E}[ | \mathbb{E}[\zeta_1 X_{t+1} + \zeta_2 ( \sigma_D (G_{t+1} X_{t+1}) ) \mid \mathcal{F}_t  ] \label{eqn:paley-zygmund-satisfy-lower-bound-eq-7} \\
    & \quad + \zeta_2 V_{t+1} |  \mid \mathcal{F}_t ] \\
    &= \mathbb{E}\left[ \abs{B_{t+1}(\zeta)} \mid \mathcal{F}_t \right] \label{eqn:paley-zygmund-satisfy-lower-bound-eq-7b}
\end{align}
where \eqref{eqn:paley-zygmund-satisfy-lower-bound-ineq-2} follows from the tower property, Jensen's inequality and the monotonocity of conditional expectation, and \eqref{eqn:paley-zygmund-satisfy-lower-bound-eq-7} follows since $V_{t+1}$ is $V_{t+1},\mathcal{F}_t$-measurable and $\zeta_1 X_{t+1}+ \zeta_2 ( \sigma_D (G_{t+1} X_{t+1}) )$ is independent of $V_{t+1}$, and \eqref{eqn:paley-zygmund-satisfy-lower-bound-eq-7b} follows by defining
\begin{equation}
    B_{t+1}(\zeta):= \mathbb{E}\left[\zeta_1 X_{t+1} + \zeta_2 \left( \sigma_D (G_{t+1} X_{t+1}) \right) \mid \mathcal{F}_t  \right] + \zeta_2 V_{t+1}. \label{eqn:paley-zygmund-satisfy-lower-bound-define-B}
\end{equation}
for $t \in \mathbb{N}_0$. The lower bounds from \eqref{eqn:paley-zygmund-satisfy-lower-bound-eq-5b} and \eqref{eqn:paley-zygmund-satisfy-lower-bound-eq-7b} are then combined to obtain $\mathbb{E}\left[ \abs{\zeta^{\top} Z_{t+1}} \mid \mathcal{F}_t \right] \geq \max\left(\mathbb{E}\left[ \abs{A_{t+1}(\zeta)} \mid \mathcal{F}_t \right],\mathbb{E}\left[ \abs{B_{t+1}(\zeta)} \mid \mathcal{F}_t \right]\right)$.

From Lemma \ref{lemma:lower-bound-A}, we have that $\mathbb{E}\left[ \abs{A_t(\zeta)} \mid \mathcal{F}_t \right] \geq f(\zeta_1,\zeta_2)$, where
\begin{align}
    f(\zeta_1,\zeta_2) := \begin{cases}
        \abs{\zeta_1}\left( \Sigma_W \sqrt{\frac{2}{\pi}} \right), \quad \zeta_1 \neq 0, \zeta_2 = 0 \\
        0, \quad \zeta_1 = 0, \zeta_2 \in \mathbb{R}\\
        \expf{-\frac{D^2 \zeta_2^2}{2 \Sigma_W^2 \zeta_1^2}}\sqrt{\frac{2}{\pi}}\Sigma_W \abs{\zeta_1} \\
        \qquad + D \zeta_2 \left( 1 + \erff{\frac{D \zeta_2}{\sqrt{2}\abs{\zeta_1}}} \right), \\
        \qquad \zeta_1 \neq 0, \zeta_2 < 0\\
        \expf{-\frac{D^2 \zeta_2^2}{2 \Sigma_W^2 \zeta_1^2}}\sqrt{\frac{2}{\pi}}\Sigma_W \abs{\zeta_1} \\
        \qquad - D \zeta_2 \left( \erfcf{\frac{D \zeta_2}{\sqrt{2}\abs{\zeta_1}}} \right), \\
        \qquad \zeta_1 \neq 0, \zeta_2 > 0
    \end{cases}. \label{eqn:definition-f}
\end{align}

From Lemma \ref{lemma:lower-bound-B}, it follows that $\mathbb{E}\left[ \abs{B_t(\zeta)} \mid \mathcal{F}_t \right] \geq \frac{\abs{\zeta_2}C}{2}$. Let us define $g(\zeta_2):= \frac{\abs{\zeta_2}C}{2}$.

Observe that $\max\left(\mathbb{E}\left[ \abs{A_t(\zeta)} \mid \mathcal{F}_t \right],\mathbb{E}\left[ \abs{B_t(\zeta)} \mid \mathcal{F}_t \right]\right) \geq \max \left( f(\zeta_1,\zeta_2), g(\zeta_2) \right)$. Now, we aim to prove that there exists $\gamma > 0$ such that for all $\zeta=(\zeta_1,\zeta_2) \in \mathcal{S}^1$ and $t \geq 0$, $\max \left( f(\zeta_1,\zeta_2), g(\zeta_2) \right) \geq \gamma$. In order to do so, let us parameterize $\zeta$ by the angle $\phi$. Specifically, we let $\zeta = (\zeta_1,\zeta_2) = (\cos(\phi), \sin(\phi))$, and then we will prove that for all $\phi \in [-\pi , \pi]$, $\max \left( f(\cos (\phi), \sin (\phi)), g(\sin (\phi)) \right) \geq \gamma$.

To aid in this proof, note that $g(\sin(\cdot))$ exhibits the following useful properties: 1) $g(\sin(\phi))$ is continuous over $\phi \in [-\pi, \pi]$; 2) $g(\sin(\phi))$ is strictly increasing over $\phi \in [-\pi, -\pi/2]$ and $[0,\pi/2]$, and strictly decreasing over $[-\pi/2,0]$ and $[\pi/2, \pi]$; 3) $g(\sin(\phi)) \geq 0$ for all $\phi \in [-\pi,\pi]$.

Consider the case where $\phi \in [-\pi, -\pi/2]$. Using Lemma \ref{lemma:properties-f} and the properties of $g(\sin(\cdot))$, $f(\cos (\phi), \sin (\phi))$ and $g(\sin (\phi))$ satisfy: 1) $g(\sin(\phi))$, $f(\cos (\phi), \sin (\phi))$ are continuous over $[-\pi, -\pi/2]$; 2) $ f(\cos (-\pi), \sin (-\pi)) = \Sigma_W \sqrt{\frac{2}{\pi}} > g(\sin(-\pi)) = 0$ and $g(\sin(-\pi/2)) = \frac{C}{2} > f(\cos (-\pi/2), \sin (-\pi/2)) = 0 $; 3) $g(\sin (\phi))$ is strictly increasing and $f(\cos (\phi), \sin (\phi))$ is strictly decreasing over $\phi \in [-\pi,-\pi/2]$; 4) $\max({g(\sin (-\pi)),f(\cos (-\pi/2), \sin (-\pi/2))}) \geq 0$. Therefore, there exists $\gamma_1 > 0$ such that for all $\phi \in [-\pi, -\pi/2]$, $\max(f(\cos (\phi), \sin (\phi)),g(\sin (\phi))) \geq \gamma_1$. Using a similar argument, the same holds true when $\phi \in [0,\pi/2]$, such that there exists $\gamma_3 > 0$ such that for all $\phi \in [0,\pi/2]$, $\max(f(\cos (\phi), \sin (\phi)),g(\sin (\phi))) \geq \gamma_3$.

Next, consider the case where $\phi \in [-\pi/2, 0]$. Using Lemma \ref{lemma:properties-f} and the properties of $g(\sin(\cdot))$, $f(\cos (\phi), \sin (\phi))$ and $g(\sin (\phi))$ satisfy: 1) $f(\cos (\phi), \sin (\phi))$, $g(\sin(\phi))$ are continuous over $[-\pi/2, 0]$; 2) $g(\sin(-\pi/2)) = \frac{C}{2} >f(\cos (-\pi/2), \sin (-\pi/2)) = 0$ and $f(\cos (0), \sin (0)) = \Sigma_w \sqrt{\frac{2}{\pi}}>g(\sin(0)) = 0$; 3) $f(\cos (\phi), \sin (\phi))$ is strictly increasing and $g(\sin (\phi))$ is strictly decreasing over $\phi \in [-\pi/2,0]$; 4) $\max({f(\cos (-\pi/2), \sin (-\pi/2)),g(\sin (0))}) \geq 0$. Therefore, there exists $\gamma_2 > 0$ such that for all $\phi \in [-\pi/2, 0]$, $\max(f(\cos (\phi), \sin (\phi)),g(\sin (\phi))) \geq \gamma_2$. Using a similar argument, the same holds true when $\phi \in [\pi / 2, \pi]$, such that there exists $\gamma_4 > 0$ such that for all $\phi \in [\pi / 2, \pi]$, $\max(f(\cos (\phi), \sin (\phi)),g(\sin (\phi))) \geq \gamma_4$.

Setting $\gamma = \min (\gamma_1,\gamma_2,\gamma_3,\gamma_4)$, the conclusion then follows.

\end{proof}

\begin{lemma} \label{lemma:lower-bound-A}
Suppose A1-A2 hold on the closed-loop system \eqref{eqn:closed-loop-system-v2} \ssnote{and $x_0 \in \mathbb{R}$}. Let $(\mathcal{F}_t)_{t \in \mathbb{N}_0}$ be the natural filtration of $(Z_t)_{t \in \mathbb{N}_0}$ from \eqref{eqn:state-input-data}. Then, $\mathbb{E}\left[ \abs{A_{t+1}(\zeta)} \mid \mathcal{F}_t \right] \geq f(\zeta_1,\zeta_2)$ for all $\zeta \in \mathcal{S}^1$ and $t \in \mathbb{N}_0$, where $A_{t+1}$ is from \eqref{eqn:paley-zygmund-satisfy-lower-bound-define-A} and $f$ is from \eqref{eqn:definition-f}.
\end{lemma}

\begin{proof}
Suppose $t \geq 0$, and $\zeta = (\zeta_1,\zeta_2) \in \mathbb{R}^2$. Define a new random variable $Q_{t+1}$ taking values in $\mathbb{R}$, satisfying
\begin{align}
    Q_{t+1} \in \argmin_{q \in [-D,D]} \abs{\zeta_1 X_{t+1} + \zeta_2 q}.
\end{align}
When $\zeta_1 \neq 0$ and $\zeta_2 \neq 0$, $Q_{t+1}$ satisfies
\begin{align}
    Q_{t+1} &= \begin{cases}
        -\frac{\zeta_1}{\zeta_2}X_{t+1}, \quad -D \leq -\frac{\zeta_1}{\zeta_2}X_{t+1} \leq D\\
        -D, \quad -\frac{\zeta_1}{\zeta_2}X_{t+1} < -D\\
        D, \quad -\frac{\zeta_1}{\zeta_2}X_{t+1} > D
    \end{cases}.
\end{align}
When $\zeta_1 = 0$ or $\zeta_2 = 0$, a satisfactory choice of $Q_{t+1}$ is $Q_{t+1} = 0$.

Since $\sigma_D(G_{t+1} X_{t+1})$ takes values in $[-D,D]$, it follows that
\begin{align}
    &\mathbb{E}\left[ \abs{  \zeta_1 X_{t+1} + \zeta_2 \left( \sigma_D (G_{t+1} X_{t+1}) \right) } \mid \mathcal{F}_t \right] \\
    &\geq \mathbb{E}\left[ \abs{  \zeta_1 X_{t+1} + \zeta_2 Q_{t+1} } \mid \mathcal{F}_t \right].
\end{align}

When $\zeta_1 = 0$ and $\zeta_2 \in \mathbb{R}$, $ \abs{  \zeta_1 X_{t+1} + \zeta_2 Q_{t+1} } = 0$. Thus, $\mathbb{E}\left[ \abs{  \zeta_1 X_{t+1} + \zeta_2 Q_{t+1} } \mid \mathcal{F}_t \right] = 0$.

When $\zeta_1 \neq 0$ and $\zeta_2 = 0$, $\abs{  \zeta_1 X_{t+1} + \zeta_2 Q_{t+1} } = \abs{\zeta_1 X_{t+1}}$. It follows that 
\begin{align}
    &\mathbb{E}[ \abs{  \zeta_1 X_{t+1} + \zeta_2 Q_{t+1} } \mid \mathcal{F}_t ] \\
    &= \mathbb{E}[ \abs{\zeta_1 X_{t+1}} \mid \mathcal{F}_t ] \\
    &= \abs{\zeta_1} \mathbb{E}[ \abs{X_{t+1}} \mid \mathcal{F}_t ] \\
    &= \abs{\zeta_1}\Bigg( \Sigma_W \sqrt{\frac{2}{\pi}} \expf{-\frac{(\theta_*^{\top}Z_t)^2}{2 \Sigma_W^2}} \label{eqn:w1-neq-0-w2-0-folded-normal} \\
    & \quad + ( \theta_*^{\top} Z_t ) \erff{\frac{\theta_*^{\top}Z_t}{\sqrt{2 \Sigma_W^2}}}  \Bigg) \\
    & \geq \abs{\zeta_1} \left( \Sigma_W \sqrt{\frac{2}{\pi}} \right). \label{eqn:w1-neq-0-w2-0-min}
\end{align}
where \eqref{eqn:w1-neq-0-w2-0-folded-normal} is due to the fact that $X_{t+1} = \theta_*^{\top}Z_t + W_t$, so $X_{t+1} \mid {Z_t=z_t} \sim \mathcal{N}\left( \theta_*^{\top}z_t, \Sigma_W^2 \right)$, and hence $\mathbb{E}\left[ \abs{X_{t+1}} \mid \mathcal{F}_t \right]=\mathbb{E}\left[ \abs{X_{t+1}} \mid Z_t \right]$ is the mean of the corresponding folded normal distribution. \eqref{eqn:w1-neq-0-w2-0-min} follows since \eqref{eqn:w1-neq-0-w2-0-folded-normal} is minimised at $\theta_*^{\top}Z_t = 0$.

When $\zeta_1 \neq 0$ and $\zeta_2 \neq 0$, we have
\begin{align}
    \abs{  \zeta_1 X_{t+1} + \zeta_2 Q_{t+1} } = \begin{cases}
        0, \quad \abs{-\frac{\zeta_1}{\zeta_2}X_{t+1}} \leq D, \\
        \abs{\zeta_1 X_{t+1} - \zeta_2 D}, \\
        \quad -\frac{\zeta_1}{\zeta_2}X_{t+1} < -D, \\
        \abs{\zeta_1 X_{t+1} + \zeta_2 D}, \\
        \quad -\frac{\zeta_1}{\zeta_2}X_{t+1} > D
    \end{cases}.
\end{align}
The conditional expectation is then given by
\begin{align}
    &\mathbb{E}\left[ \abs{  \zeta_1 X_{t+1} + \zeta_2 Q_{t+1} } \mid \mathcal{F}_t \right] \allowdisplaybreaks \\
    &= \mathbb{E}\Bigg[ \bm{1}_{ \left\{ \abs{-\frac{\zeta_1}{\zeta_2}X_{t+1}} \leq D \right\}} \abs{  \zeta_1 X_{t+1} + \zeta_2 Q_{t+1} } \allowdisplaybreaks \\
    &\qquad + \bm{1}_{ \left\{ -\frac{\zeta_1}{\zeta_2}X_{t+1} < -D \right\} } \abs{  \zeta_1 X_{t+1} + \zeta_2 Q_{t+1} } \allowdisplaybreaks \\
    &\qquad + \bm{1}_{ \left\{ -\frac{\zeta_1}{\zeta_2}X_{t+1} > D \right\} } \abs{  \zeta_1 X_{t+1} + \zeta_2 Q_{t+1} } \mid \mathcal{F}_t \Bigg] \allowdisplaybreaks \\
    &= \mathbb{E}\left[ \bm{1}_{ \left\{ -\frac{\zeta_1}{\zeta_2}X_{t+1} < -D \right\} } \abs{\zeta_1 X_{t+1} - \zeta_2 D} \mid \mathcal{F}_t \right] \label{eqn:expec-w1-w2-g-0-1} \allowdisplaybreaks \\
    &\qquad +  \mathbb{E}\left[ \bm{1}_{ \left\{ -\frac{\zeta_1}{\zeta_2}X_{t+1} > D \right\} } \abs{\zeta_1 X_{t+1} + \zeta_2 D} \mid \mathcal{F}_t \right] \label{eqn:expec-w1-w2-neq-0-2}
\end{align}
We further split $\zeta_2 \neq 0$ into two cases, where $\zeta_2 < 0$, and $\zeta_2 > 0$.
Let us start with $\zeta_1 \neq 0$ and $\zeta_2 < 0$.
Evaluating the conditional expectation in \eqref{eqn:expec-w1-w2-g-0-1} and \eqref{eqn:expec-w1-w2-neq-0-2}, we have
\begin{align}
    &\mathbb{E}\left[ \bm{1}_{ \left\{ -\frac{\zeta_1}{\zeta_2}X_{t+1} < -D \right\} } |\zeta_1 X_{t+1} - \zeta_2 D| \mid \mathcal{F}_t \right] \allowdisplaybreaks \\ 
    &\quad + \mathbb{E}\left[ \bm{1}_{ \left\{ -\frac{\zeta_1}{\zeta_2}X_{t+1} > D \right\} } |\zeta_1 X_{t+1} + \zeta_2 D| \mid \mathcal{F}_t \right] \allowdisplaybreaks \\
    &= -\mathbb{E}\left[ \bm{1}_{ \left\{ \zeta_1 X_{t+1} - \zeta_2 D < 0 \right\} } \left({\zeta_1 X_{t+1} - \zeta_2 D}\right) \mid \mathcal{F}_t \right] \\
    &\quad +\mathbb{E}[ \bm{1}_{ \{ \zeta_1 X_{t+1} + \zeta_2 D > 0 \} } ({\zeta_1 X_{t+1} + \zeta_2 D}) \mid \mathcal{F}_t ] \\ \allowdisplaybreaks
    &= \frac{\Sigma_W | \zeta_1 |}{\sqrt{2 \pi}}\bigg( \expf{-\frac{(D \zeta_2 - \zeta_1 \theta_*^{\top}Z_t)^2}{2 \Sigma_W^2 \zeta_1^2}}  \\ \allowdisplaybreaks 
    & \quad + \expf{-\frac{(D \zeta_2 + \zeta_1 \theta_*^{\top}Z_t)^2}{2 \Sigma_W^2 \zeta_1^2}} \bigg) \allowdisplaybreaks \\
    & \quad + \frac{1}{2} \bigg( ( D \zeta_2 + \zeta_1 \theta_*^{\top} Z_t ) \bigg( 1  \allowdisplaybreaks \\
    &\quad + \erff{\frac{D \zeta_2 + \zeta_1 \theta_*^{\top}Z_t}{\sqrt{2} \Sigma_W |\zeta_1|}}\bigg) \allowdisplaybreaks \\
    & \quad + ( D \zeta_2 - \zeta_1 \theta_*^{\top} Z_t ) \erfcf{\frac{-D \zeta_2 + \zeta_1 \theta_*^{\top}Z_t}{\sqrt{2} \Sigma_W |\zeta_1|}} \bigg) \label{eqn:lower-bound-A-eq-1} \allowdisplaybreaks \\
    & =: h_1(\zeta_1,\zeta_2,\theta_*^{\top}Z_t) \label{eq:lower-bound-A-def-h1}
\end{align}
where \eqref{eqn:lower-bound-A-eq-1} follows from the fact that $\zeta_1 X_{t+1} - \zeta_2 D \mid {Z_t=z_t} \sim \mathcal{N}\left( \zeta_1 \theta_*^{\top} z_t - \zeta_2 D, \zeta_1^2 \Sigma_W^2 \right)$ and $\zeta_1 X_{t+1} + \zeta_2 D \mid {Z_t=z_t} \sim \mathcal{N}\left( \zeta_1 \theta_*^{\top} z_t + \zeta_2 D, \zeta_1^2 \Sigma_W^2 \right)$, $\mathbb{E}[X \bm{1}_{X > a}] = \mathbb{E}[X \mid X > a] P(X > a)$ and $\mathbb{E}[X \bm{1}_{X < a}] = \mathbb{E}[X \mid X < a] P(X < a)$ for random variable $X$ taking values in $\mathbb{R}$ and $a \in \mathbb{R}$, and the conditional expectation of a truncated Gaussian distribution \cite[Theorem 22.2]{greene2002econometric}. Additionally, we denote the conditional expectation by the function $h_1$ for ease of notation.
Taking the partial derivative of $h_1$ with respect to $\hat{x}$ we have
\begin{align}
    &\frac{\partial h_1(\zeta_1,\zeta_2,\hat{x})}{\partial \hat{x}} \\
    &= \frac{1}{2} \abs{\zeta_1} \left(\erff{\frac{D \zeta_2 + \abs{\zeta_1} \hat{x}}{\sqrt{2} \Sigma_W \abs{\zeta_1}}} + \erff{\frac{-D \zeta_2 + \abs{\zeta_1} \hat{x}}{\sqrt{2} \Sigma_W \abs{\zeta_1}}} \right).
\end{align}
This partial derivative was symbolically computed using a CAS.
Next, suppose $\zeta_1 \neq 0$, and $\zeta_2 < 0$. When $\hat{x} = 0$, $\frac{\partial h_1(\zeta_1,\zeta_2,\hat{x})}{\partial \hat{x}} = 0$. When $\hat{x}>0$, $\frac{\partial h_1(\zeta_1,\zeta_2,\hat{x})}{\partial \hat{x}} > 0$. When $\hat{x}<0$, $\frac{\partial h_1(\zeta_1,\zeta_2,\hat{x})}{\partial \hat{x}} < 0$. Thus, $h_1(\zeta_2,\zeta_2,\hat{x})$ is minimised at $\hat{x}=0$. We use this to lower bound \eqref{eq:lower-bound-A-def-h1} for all $ \zeta_1 \neq 0$ and $\zeta_2 < 0$:
\begin{align}
    &h_1(\zeta_1,\zeta_2,\theta_*^{\top}Z_t) \allowdisplaybreaks \\
    & \geq \expf{-\frac{D^2 \zeta_2^2}{2 \Sigma_W^2 \zeta_1^2}}\sqrt{\frac{2}{\pi}}\Sigma_W \abs{\zeta_1} \label{eqn:lower-bound-A-ineq-1} \allowdisplaybreaks \\
    & \quad + D \zeta_2 \left( 1 + \erff{\frac{D \zeta_2}{\sqrt{2}\Sigma_W\abs{\zeta_1}}} \right)
\end{align}
Combining \eqref{eq:lower-bound-A-def-h1} and \eqref{eqn:lower-bound-A-ineq-1} we have that for all $\zeta_1 \neq 0$ and $\zeta_2 < 0$,
\begin{align}
    &\mathbb{E}\left[ \bm{1}_{ \left\{ -\frac{\zeta_1}{\zeta_2}X_{t+1} < -D \right\} } \abs{\zeta_1 X_{t+1} - \zeta_2 D} \mid \mathcal{F}_t \right] \allowdisplaybreaks \\
    & \quad +  \mathbb{E}\left[ \bm{1}_{ \left\{ -\frac{\zeta_1}{\zeta_2}X_{t+1} > D \right\} } \abs{\zeta_1 X_{t+1} + \zeta_2 D} \mid \mathcal{F}_t \right]  \allowdisplaybreaks \\
    &\geq \expf{-\frac{D^2 \zeta_2^2}{2 \Sigma_W^2 \zeta_1^2}}\sqrt{\frac{2}{\pi}}\Sigma_W \abs{\zeta_1} \label{eqn:lower-bound-A-ineq-2} \allowdisplaybreaks \\
    & \quad + D \zeta_2 \left( 1 + \erff{\frac{D \zeta_2}{\sqrt{2} \Sigma_W \abs{\zeta_1}}} \right).
\end{align}

Now, we focus on the case where $\zeta_1 \neq 0$ and $\zeta_2 > 0$. Evaluating the conditional expectation in \eqref{eqn:expec-w1-w2-g-0-1} and \eqref{eqn:expec-w1-w2-neq-0-2}, we have
\begin{align}
    &\mathbb{E}[ \bm{1}_{ \{ -\frac{\zeta_1}{\zeta_2}X_{t+1} < -D \} } \abs{\zeta_1 X_{t+1} - \zeta_2 D} \mid \mathcal{F}_t ] \allowdisplaybreaks \allowdisplaybreaks \\
    & \quad + \mathbb{E}[ \bm{1}_{ \{ -\frac{\zeta_1}{\zeta_2}X_{t+1} > D \} } \abs{\zeta_1 X_{t+1} + \zeta_2 D} \mid \mathcal{F}_t ] \allowdisplaybreaks \allowdisplaybreaks \\
    &= \mathbb{E}[ \bm{1}_{ \{ 0 < \zeta_1 X_{t+1} - \zeta_2 D \} } ({\zeta_1 X_{t+1} - \zeta_2 D}) \mid \mathcal{F}_t ] \allowdisplaybreaks \allowdisplaybreaks \\
    & \quad -\mathbb{E}[ \bm{1}_{ \{ 0 > \zeta_1 X_{t+1} + \zeta_2 D \} } ({\zeta_1 X_{t+1} + \zeta_2 D}) \mid \mathcal{F}_t ] \allowdisplaybreaks \allowdisplaybreaks \\
    &= \frac{\Sigma_W \abs{\zeta_1}}{\sqrt{2 \pi}} \bigg( \expf{-\frac{(D \zeta_2 - \zeta_1 \theta_*^{\top}Z_t)^2}{2 \Sigma_W^2 \zeta_1^2}} \allowdisplaybreaks \\
    & \quad + \expf{-\frac{(D \zeta_2 + \zeta_1 \theta_*^{\top}Z_t)^2}{2 \Sigma_W^2 \zeta_1^2}} \bigg) \allowdisplaybreaks \allowdisplaybreaks \\
    & \quad + \frac{1}{2} \bigg( ( D \zeta_2 - \zeta_1 \theta_*^{\top} Z_t ) \bigg( -2 \allowdisplaybreaks \allowdisplaybreaks \\
    & \quad + \erfcf{\frac{-D \zeta_2 + \zeta_1 \theta_*^{\top}Z_t}{\sqrt{2} \Sigma_W \abs{\zeta_1}}} \bigg) \allowdisplaybreaks \allowdisplaybreaks \\
    & \quad - ( D \zeta_2 + \zeta_1 \theta_*^{\top} Z_t ) \erfcf{\frac{D \zeta_2 + \zeta_1 \theta_*^{\top}Z_t}{\sqrt{2} \Sigma_W \abs{\zeta_1}}} \bigg) \allowdisplaybreaks \label{eqn:lower-bound-A-eq-2} \\
    & =: h_2(\zeta_1,\zeta_2,\theta_*^{\top}Z_t) \label{eq:lower-bound-A-def-h2}
\end{align}
where \eqref{eqn:lower-bound-A-eq-2} follows similarly to \eqref{eq:lower-bound-A-def-h1} using the fact that $\zeta_1 X_{t+1} - \zeta_2 D \mid {Z_t=z_t} \sim \mathcal{N}\left( \zeta_1 \theta_*^{\top} z_t - \zeta_2 D, \zeta_1^2 \Sigma_W^2 \right)$ and $\zeta_1 X_{t+1} + \zeta_2 D \mid {Z_t=z_t} \sim \mathcal{N}\left( \zeta_1 \theta_*^{\top} z_t + \zeta_2 D, \zeta_1^2 \Sigma_W^2 \right)$. Additionally, we denote the conditional expectation by the function $h_2$ for ease of notation.
Taking the partial derivative of $h_2$ with respect to $\hat{x}$, we arrive at
\begin{align}
    &\frac{\partial h_2(\zeta_1,\zeta_2,\hat{x})}{\partial \hat{x}} \\
    &= \frac{1}{2}\abs{\zeta_1}\left( \erff{\frac{-D\zeta_2 + \abs{\zeta_1} \hat{x}}{\sqrt{2}\Sigma_W \abs{\zeta_1}}} + \erff{\frac{D \zeta_2 + \abs{\zeta_1} \hat{x}}{\sqrt{2}\Sigma_W \abs{\zeta_1}}} \right)
\end{align}
Suppose $\zeta_1 \neq 0$, and $\zeta_2 > 0$. When $\hat{x} = 0$, $\frac{\partial h_2(\zeta_1,\zeta_2,\hat{x})}{\partial \hat{x}} = 0$. When $\hat{x}>0$, $\frac{\partial h_2(\zeta_1,\zeta_2,\hat{x})}{\partial \hat{x}} > 0$. When $\hat{x}<0$, $\frac{\partial h_2(\zeta_1,\zeta_2,\hat{x})}{\partial \hat{x}} < 0$. Thus, $h_2(\zeta_2,\zeta_2,\hat{x})$ is minimised at $\hat{x}=0$. We use this to lower bound \eqref{eq:lower-bound-A-def-h2} for all $\zeta_1 \neq 0$, $\zeta_2 > 0$ and $\hat{x} \in \mathbb{R}$:
\begin{align}
    &h_2(\zeta_1,\zeta_2,\theta_*^{\top}Z_t) \allowdisplaybreaks \\ 
    &\geq \expf{-\frac{D^2\zeta_2^2}{2 \Sigma_W^2 \zeta_1^2}} \sqrt{\frac{2}{\pi}} \Sigma_W \abs{\zeta_1} - D \zeta_2 \erfcf{\frac{D \zeta_2}{\sqrt{2} \Sigma_W \abs{\zeta_1}}}
\end{align}
Therefore, for all $\zeta_1 \neq 0$ and $\zeta_2 > 0$, the following holds.
\begin{align}
    &\mathbb{E}\left[ \bm{1}_{ \left\{ -\frac{\zeta_1}{\zeta_2}X_{t+1} < -D \right\} } \abs{\zeta_1 X_{t+1} - \zeta_2 D} \mid \mathcal{F}_t \right] \allowdisplaybreaks \\
    &\quad +  \mathbb{E}\left[ \bm{1}_{ \left\{ -\frac{\zeta_1}{\zeta_2}X_{t+1} > D \right\} } \abs{\zeta_1 X_{t+1} + \zeta_2 D} \mid \mathcal{F}_t \right]  \allowdisplaybreaks \\
    & \geq \expf{-\frac{D^2 \zeta_2^2}{2 \Sigma_W^2 \zeta_1^2}}\sqrt{\frac{2}{\pi}}\Sigma_W \abs{\zeta_1} \allowdisplaybreaks \\
    &\quad - D \zeta_2 \left( \erfcf{\frac{D \zeta_2}{\sqrt{2} \Sigma_W \abs{\zeta_1}}} \right)
\end{align}

The conclusion follows by observing that $\mathbb{E}\left[ \abs{A_{t+1}(\zeta)} \mid \mathcal{F}_t \right] \geq f(\zeta_1,\zeta_2)$ holds for all $(\zeta_1,\zeta_2) \in \mathbb{R}^2$.
\end{proof}

\begin{lemma} \label{lemma:lower-bound-B}
Let $(\mathcal{F}_t)_{t \in \mathbb{N}_0}$ be the natural filtration of $(Z_t)_{t \in \mathbb{N}_0}$ from \eqref{eqn:state-input-data}. Suppose A1-A2 hold on the closed-loop system \eqref{eqn:closed-loop-system-v2} and $x_0 \in \mathbb{R}$. Then, $\mathbb{E}\left[ \abs{B_{t+1}(\zeta)} \mid \mathcal{F}_t \right] \geq \frac{\abs{\zeta_2}C}{2}$, where $B_{t+1}$ is from \eqref{eqn:paley-zygmund-satisfy-lower-bound-define-B}.
\end{lemma}

\begin{proof}
Suppose $\zeta = (\zeta_1,\zeta_2) \in \mathbb{R}^2$ and $t \in \mathbb{N}_0$. When $\zeta_2=0$, using the monotonocity of conditional expectation we have
\begin{align}
    &\mathbb{E}\left[ \abs{B_{t+1}(\zeta)} \mid \mathcal{F}_t \right] \geq 0.
\end{align}
Next, let $g_t$ be the mapping satisfying $g_t(\zeta_1,\zeta_2,(Z_0,\hdots,Z_t))=\mathbb{E}\left[ \zeta_1 X_{t+1} + \zeta_2 \sigma_D (G_{t+1} X_{t+1}) \mid \mathcal{F}_t \right]$. Note that the distribution of $B_{t+1}(\zeta) \mid {Z_0=z_0},\hdots,{Z_t=z_t}$ is $\text{Uniform}( [ g_t(\zeta_1,\zeta_2,(z_0,\hdots,z_t)) - \abs{\zeta_2 C}, g_t(\zeta_1,\zeta_2,(z_0,\hdots,z_t)) + \abs{\zeta_2 C} ] )$ for all $z_0,\hdots,z_t \in \mathbb{R}^2$. Using this fact and the law of the unconscious statistician, when $\zeta_2 \neq 0$, we have
\begin{align}
    & \mathbb{E}\left[ \abs{B_{t+1}(\zeta)} \mid Z_0=z_0,\hdots,Z_t=z_t \right] \allowdisplaybreaks \\
    &= \begin{cases}
        \frac{1}{2 \abs{\zeta_2} C}\int_{g_t(\zeta_1,\zeta_2,(z_0,\hdots,z_t)) - \abs{\zeta_2}C}^{g_t(\zeta_1,\zeta_2,(z_0,\hdots,z_t)) + \abs{\zeta_2}C} b \ db , \allowdisplaybreaks \\
        \quad g_t(\zeta_1,\zeta_2,(z_0,\hdots,z_t)) > \abs{\zeta_2}C \allowdisplaybreaks \\
        -\frac{1}{2 \abs{\zeta_2} C}\int_{g_t(\zeta_1,\zeta_2,(z_0,\hdots,z_t)) - \abs{\zeta_2}C}^{g_t(\zeta_1,\zeta_2,(z_0,\hdots,z_t)) + \abs{\zeta_2}C} b \ db , \allowdisplaybreaks \\
        \quad g_t(\zeta_1,\zeta_2,(z_0,\hdots,z_t)) < -\abs{\zeta_2}C \allowdisplaybreaks \\
        \frac{1}{2 \abs{\zeta_2} C} \Big( -\int_{g_t(\zeta_1,\zeta_2,(z_0,\hdots,z_t)) - \abs{\zeta_2}C}^{0} b \ db \allowdisplaybreaks \\
        \quad + \int_{0}^{g_t(\zeta_1,\zeta_2,(z_0,\hdots,z_t)) + \abs{\zeta_2}C} b \ db \Big), \allowdisplaybreaks \\
        \quad  -\abs{\zeta_2}C \leq g_t(\zeta_1,\zeta_2,(z_0,\hdots,z_t)) \leq \abs{\zeta_2}C
    \end{cases} \allowdisplaybreaks \\
    & =\begin{cases}
        g_t(\zeta_1,\zeta_2,(z_0,\hdots,z_t)), \allowdisplaybreaks \\
        \quad g_t(\zeta_1,\zeta_2,(z_0,\hdots,z_t)) > \abs{\zeta_2}C \allowdisplaybreaks \\
        -g_t(\zeta_1,\zeta_2,(z_0,\hdots,z_t)), \allowdisplaybreaks \\
        \quad g_t(\zeta_1,\zeta_2,(z_0,\hdots,z_t)) < -\abs{\zeta_2}C \allowdisplaybreaks \\
        \frac{1}{2 \abs{\zeta_2} C} \left[ g_t(\zeta_1,\zeta_2,(z_0,\hdots,z_t))^2 + \zeta_2^2 C^2 \right], \allowdisplaybreaks \\ 
        \quad  -\abs{\zeta_2}C \leq g_t(\zeta_1,\zeta_2,(z_0,\hdots,z_t)) \leq \abs{\zeta_2}C
    \end{cases} \allowdisplaybreaks \\
    & \geq \begin{cases}
        \abs{\zeta_2}C, \quad g_t(\zeta_1,\zeta_2,(z_0,\hdots,z_t)) > \abs{\zeta_2}C \allowdisplaybreaks \\
        \abs{\zeta_2}C, \quad g_t(\zeta_1,\zeta_2,(z_0,\hdots,z_t)) < -\abs{\zeta_2}C \allowdisplaybreaks \\
        \frac{\abs{\zeta_2}C}{2}, \allowdisplaybreaks \\
        \quad  -\abs{\zeta_2}C \leq g_t(\zeta_1,\zeta_2,(z_0,\hdots,z_t)) \leq \abs{\zeta_2}C
    \end{cases} \allowdisplaybreaks \\
    & \geq \frac{\abs{\zeta_2}C}{2}
\end{align}
for all $z_0,\hdots,z_t \in \mathbb{R}^2$. Thus, for all $(\zeta_1,\zeta_2) \in \mathbb{R}$, we have $\mathbb{E}\left[ \abs{ B_{t+1}(\zeta) }  \mid \mathcal{F}_t \right] = \mathbb{E}\left[ \abs{ B_{t+1}(\zeta) }  \mid Z_0,\hdots,Z_t \right] \geq \frac{\abs{\zeta_2}C}{2}$.
\end{proof}

\begin{lemma} \label{lemma:properties-f}
Consider the function f from \eqref{eqn:definition-f}. The following properties hold:
\begin{enumerate}
    \item $f(\cos (\phi), \sin(\phi))$ is continuous over $\phi \in [-\pi, \pi]$;
    \item $f(\cos (\phi), \sin(\phi))$ is strictly decreasing over $\phi \in [-\pi, -\pi/2]$ and $[0,\pi/2]$, and strictly increasing over $[-\pi/2,0]$ and $[\pi/2, \pi]$;
    \item $f(\cos(\phi),\sin(\phi)) \geq 0$ for all $\phi \in [-\pi,\pi]$.
\end{enumerate}
\end{lemma}
\begin{proof}

The proof of 1) follows from the fact that $f$ is continuous over the domains $\{ (\zeta_1,\zeta_2) \in \mathbb{R}^2 \mid \zeta_1 \neq 0, \zeta_2 = 0 \}$, $\{ (\zeta_1,\zeta_2) \in \mathbb{R}^2 \mid \zeta_1 \in 0, \zeta_2 \in \mathbb{R} \}$, $\{ (\zeta_1,\zeta_2) \in \mathbb{R}^2 \mid \zeta_1 \neq 0, \zeta_2 < 0 \}$ and $\{ (\zeta_1,\zeta_2) \in \mathbb{R}^2 \mid \zeta_1 \neq 0, \zeta_2 > 0 \}$, alongside the fact that $\lim_{\phi \rightarrow c} f(\cos(\phi),\sin(\phi)) = f(\cos(c),\sin(c))$ for all $c \in \{ -\pi, -\pi/2, 0, \pi/2, \pi \}$.

We now prove 2). Over the interval $\phi \in (-\pi,-\pi/2) \cup (-\pi/2,0)$, we find that $\frac{d}{d \phi}f(\cos(\phi),\sin(\phi))=D \cos (\phi ) \left(\text{erf}\left(\frac{D \sin (\phi )}{\Sigma_W  \sqrt{\cos (2 \phi )+1}}\right)+1\right)-\frac{\Sigma_W  \sqrt{\cos (2 \phi )+1} \tan (\phi ) e^{-\frac{D^2 \tan ^2(\phi )}{2 \Sigma_W ^2}}}{\sqrt{\pi }}$ using a CAS. On the interval $\phi \in (-\pi,-\pi/2)$, $\frac{d}{d \phi}f(\cos(\phi),\sin(\phi)) < 0$ holds, and on the interval $\phi \in (-\pi/2,0)$, $\frac{d}{d \phi}f(\cos(\phi),\sin(\phi)) > 0$ holds. Thus, $f(\cos(\phi),\sin(\phi))$ is strictly decreasing and strictly increasing over the open intervals $\phi \in (-\pi,-\pi/2)$ and $\phi \in (-\pi/2,0)$ respectively, and due to the continuity of $f(\cos(\phi),\sin(\phi))$ these same properties hold over their respective associated closed intervals.
Similarly, over the interval $\phi \in (0,\pi/2) \cup (\pi/2,\pi)$, we find that $\frac{d}{d \phi}f(\cos(\phi),\sin(\phi))=-\frac{\Sigma_W  \sqrt{\cos (2 \phi )+1} \tan (\phi ) e^{-\frac{D^2 \tan ^2(\phi )}{2 \Sigma_W ^2}}}{\sqrt{\pi }}-D \cos (\phi ) \text{erfc}\left(\frac{D \sin (\phi )}{\Sigma_W  \sqrt{\cos (2 \phi )+1}}\right)$ using a CAS. On the interval $\phi \in (0,\pi/2)$, $\frac{d}{d \phi}f(\cos(\phi),\sin(\phi)) < 0$ holds, and on the interval $\phi \in (\pi/2,\pi)$, $\frac{d}{d \phi}f(\cos(\phi),\sin(\phi)) > 0$ holds. Thus, $f(\cos(\phi),\sin(\phi))$ is strictly decreasing and strictly increasing over the open intervals $\phi \in (0,\pi/2)$ and $\phi \in (\pi/2,\pi)$ respectively, and due to the continuity of $f(\cos(\phi),\sin(\phi))$ these same properties hold over their respective closed intervals. 
The proof of 2) is thus completed.

The proof of 3) follows from the fact that $f(\zeta_1,\zeta_2) \geq 0$ for all $\zeta_1,\zeta_2 \in \mathbb{R}$.
\end{proof}

\subsection{Analysis for Lemma \ref{lemma:parameter-estimate-bound}}
We provide the proof of Lemma \ref{lemma:parameter-estimate-bound}.
\begin{proof}[Proof of Lemma \ref{lemma:parameter-estimate-bound}]
    Suppose $d \in \Big(0,\frac{ 90 \Sigma_W }{\sqrt{10 \lambda_{\text{min}} \left( \Gamma_{\text{sb}} \right)}}\Big)$, and $i \geq M(d,p,q,\Gamma_{\text{sb}})$. Set
    \begin{equation}
        \delta = i^{\frac{4}{3}} e^{-i \frac{\lambda_{\text{min}} \left( \Gamma_{\text{sb}} \right) d^2 p^2}{3 ( 90 \Sigma_W )^2}} c_1(q,\Gamma_{\text{sb}})^{\frac{2}{3}} e^{\frac{1}{3}\left( c_2(p) + \log \left( \det \left( \Gamma_{\text{sb}}^{-1} \right) \right) \right)}, \label{eqn:parameter-estimation-error-tail-bound-ineq-7}
    \end{equation}
    and $\overline{\Gamma}=\frac{1}{\delta}i^2 c_1(q,\Gamma_{\text{sb}}) I$.
    
    Now, we will establish that the sequence $(Z_t,X_t)_{t=1}^i$ satisfies the premise of Proposition \ref{prop:estim-bound}. Let $\mathcal{F}_t$ be the sigma-algebra generated by $W_0,\hdots,W_t,Z_1,\hdots,Z_t$ for $t \in \mathbb{N}$. Note that (a) is satisfied since $X_{t+1} = \theta_*^{\top}Z_t + W_t$ holds for $t \in \mathbb{N}$, with $W_t \mid \mathcal{F}_{t-1} \sim \mathcal{N}(0,\Sigma_W^2)$ due to A1. Moreover, $Z_1,\hdots,Z_i$ satisfies the $(1,\Gamma_{\text{sb}},p)$-BMSB condition due to 1) in the premise, establishing (b).  We are left to prove that (c) $P\left( \sum_{t=1}^{i} Z_t Z_t^{\top} \not \preceq \overline{\Gamma}i \right) \leq \delta$, $\delta \in (0,1)$, $\Gamma_{\text{sb}} \preceq \overline{\Gamma}$, and \eqref{eqn:estim-bound-i-condition} holds with $k=1$.
    
    We start by proving that $\delta \in (0,1)$. Since $i \geq M(d,p,q,\Gamma_{\text{sb}})$, then $i \geq M'(d,p,q,\Gamma_{\text{sb}})$ holds, and so $i^{\frac{4}{3}} e^{-i \frac{\lambda_{\text{min}} \left( \Gamma_{\text{sb}} \right) d^2 p^2}{3 ( 90 \Sigma_W )^2}} c_1(q,\Gamma_{\text{sb}})^{\frac{2}{3}} e^{\frac{1}{3}\left( c_2(p) + \log \left( \det \left( \Gamma_{\text{sb}}^{-1} \right) \right) \right)} < 1$ holds by definition in \eqref{eqn:M-definition}. Thus, $\delta = i^{\frac{4}{3}} e^{-i \frac{\lambda_{\text{min}} \left( \Gamma_{\text{sb}} \right) d^2 p^2}{3 ( 90 \Sigma_W )^2}} c_1(q,\Gamma_{\text{sb}})^{\frac{2}{3}} e^{\frac{1}{3}\left( c_2(p) + \log \left( \det \left( \Gamma_{\text{sb}}^{-1} \right) \right) \right)} \in (0,1)$.

    Next, since $\delta \in (0,1)$ and $i \geq 1$, we find that $\overline{\Gamma}=\frac{1}{\delta} i^2 (q + \lambda_{\text{max}}(\Gamma_{\text{sb}}))I \succ (q + \lambda_{\text{max}}(\Gamma_{\text{sb}})) \succ \Gamma_{\text{sb}}$ holds.
    
    Next, we know from 2) in the premise that $P\left( \sum_{t=1}^{i} Z_t Z_t^{\top} \not \preceq \frac{1}{\delta}i^3 q I \right) \leq \delta$, which implies that $P\left( \sum_{t=1}^{i} Z_t Z_t^{\top} \not \preceq \overline{\Gamma}i \right) = P\left( \sum_{t=1}^{i} Z_t Z_t^{\top} \not \preceq \frac{1}{\delta}i^3 (q+\lambda_{\text{max}}(\Gamma_{\text{sb}})) I \right) \leq \delta$.
    
    We now prove that \eqref{eqn:estim-bound-i-condition} holds with $k=1$, i.e., $i \geq (10/p^2) \left( \log \left( \frac{1}{\delta} \right) + ( 4 \log \left( 10 / p \right) ) + \log \left( \det \left( \overline{\Gamma} \Gamma_{\text{sb}}^{-1} \right) \right) \right)$. Since $i \geq M(d,p,q,\Gamma_{\text{sb}})$, it follows that $i \geq \left( \frac{p^2}{10} - \frac{\lambda_{\text{min}} \left( \Gamma_{\text{sb}} \right) d^2 p^2}{ ( 90 \Sigma_W )^2} \right)^{-1}\left(  4 \log \left( 10 / p \right)  -c_2(p)  \right) $ by definition in \eqref{eqn:M-definition}. After manipulating this inequality, we find
    
    \begin{align}
        &i^{\frac{4}{3}}e^{-i\frac{\lambda_{\text{min}} \left( \Gamma_{\text{sb}} \right) d^2 p^2 }{3 ( 90 \Sigma_W )^2}} c_1(q,\Gamma_{\text{sb}})^{\frac{2}{3}} e^{\frac{1}{3}\big[ c_2(p) + \log \big( \det \big( \Gamma_{\text{sb}}^{-1} \big)\big) \big]} \nonumber \\
        & \quad \geq i^{\frac{4}{3}}e^{-i\frac{p^2}{30}} c_1(q,\Gamma_{\text{sb}})^{\frac{2}{3}} \\
        & \quad \quad  \times e^{\frac{1}{3}\big[  4 \log \left( 10 / p \right)  + \log \big( \det \big( \Gamma_{\text{sb}}^{-1} \big)\big) \big]} \label{eqn:parameter-estimation-error-tail-bound-ineq-1}
    \end{align}
    where we rely on $\left( \frac{p^2}{10} - \frac{\lambda_{\text{min}} \left( \Gamma_{\text{sb}} \right) d^2 p^2}{( 90 \Sigma_W )^2} \right) > 0$ for $d \in \Big(0,\frac{ 90 \Sigma_W }{\sqrt{10 \lambda_{\text{min}} \left( \Gamma_{\text{sb}} \right)}}\Big)$. It follows from \eqref{eqn:parameter-estimation-error-tail-bound-ineq-7} and \eqref{eqn:parameter-estimation-error-tail-bound-ineq-1} that $\delta \geq i^{\frac{4}{3}} e^{-i\frac{p^2}{30}} c_1(q,\Gamma_{\text{sb}})^{\frac{2}{3}} e^{\frac{1}{3}\left[  4 \log \left( 10 / p \right)  + \log \left( \det \left( \Gamma_{\text{sb}}^{-1} \right)\right) \right]}$. Rearranging the right hand side of this inequality, we find
    \begin{align}
        \delta &\geq e^{\frac{1}{3}\left( - \frac{i p^2}{10} + 2\log \left( i^2 c_1(q,\Gamma_{\text{sb}}) \right) +  4 \log \left( 10 / p \right)  + \log \left( \det \left( \Gamma_{\text{sb}}^{-1} \right)\right) \right)}.
    \end{align}
    Taking, the reciprocal, we have
    \begin{align}
        \frac{1}{\delta} &\leq e^{\frac{1}{3}\left( \frac{i p^2}{10} - 2\log \left( i^2 c_1(q,\Gamma_{\text{sb}}) \right) -  4 \log \left( 10 / p \right)  - \log \left( \det \left( \Gamma_{\text{sb}}^{-1} \right)\right) \right)}.
    \end{align}
    Taking the log of both sides then rearranging to isolate $i$, we arrive at
    \begin{align}
        i &\geq  \left(\frac{10}{p^2}\right) \left( \log \left( \frac{1}{\delta} \right) +  4 \log \left( \frac{10}{p} \right)  + \log\left( \det \left( \overline{\Gamma} \Gamma_{\text{sb}}^{-1} \right) \right) \right) .
    \end{align}
    
    Thus, we have satisfied the premise of Proposition \ref{prop:estim-bound}, and therefore we find that
    \begin{align}
        & P\Bigg( \norm{\hat{\theta}_i - \theta_*}_2 > \\
        & ( 90 \Sigma_W / p ) \sqrt{\frac{c_2(p) + \log \left( \det \left( \overline{\Gamma} \Gamma_{\text{sb}}^{-1} \right) \right) + \log \left( \frac{1}{\delta} \right) }{i\lambda_{\text{min}} \left( \Gamma_{\text{sb}} \right)}}  \Bigg) \leq 3 \delta. \label{eqn:parameter-estimation-error-tail-bound-ineq-6} \nonumber
    \end{align}
    
    Now, we prove that \\$\left( 90 \Sigma_W / p \right) \sqrt{\frac{c_2(p) + \log \left( \det \left( \overline{\Gamma} \Gamma_{\text{sb}}^{-1} \right) \right) + \log \left( \frac{1}{\delta} \right) }{i\lambda_{\text{min}} \left( \Gamma_{\text{sb}} \right)}} \leq d$. From \eqref{eqn:parameter-estimation-error-tail-bound-ineq-7}, we know $\delta \geq i^{\frac{4}{3}} e^{-i \frac{\lambda_{\text{min}} \left( \Gamma_{\text{sb}} \right) d^2 p^2}{3 ( 90 \Sigma_W )^2}} c_1(q,\Gamma_{\text{sb}})^{\frac{2}{3}} e^{\frac{1}{3}\left( c_2(p) + \log \left( \det \left( \Gamma_{\text{sb}}^{-1} \right) \right) \right)}$, and rearranging this inequality, we find
    \begin{align}
        \delta &\geq e^{-i\frac{\lambda_{\text{min}} \left( \Gamma_{\text{sb}} \right) d^2 p^2}{3 ( 90 \Sigma_W  )^2}+\frac{2\log \left( i^2 c_1(q,\Gamma_{\text{sb}}) \right)}{3}+\frac{1}{3}\left[ c_2(p) + \log \left( \det \left( \Gamma_{\text{sb}}^{-1} \right)\right) \right]}
    \end{align}
    Taking, the reciprocal, we have
    \begin{align}
        \frac{1}{\delta} &\leq e^{\frac{1}{3}\left( i\frac{\lambda_{\text{min}} \left( \Gamma_{\text{sb}} \right) p^2 d^2}{( 90 \Sigma_W  )^2} - 2\log \left( i^2 c_1(q,\Gamma_{\text{sb}}) \right) - c_2(p) - \log \left( \det \left( \Gamma_{\text{sb}}^{-1} \right)\right) \right)}
    \end{align}
    Taking the log of both sides then rearranging to isolate $d$, we have
    \begin{align}
        d \geq &( 90 \Sigma_W / p ) \\
        & \times \sqrt{ \frac{c_2(p) + \log \left( \det \left( \overline{\Gamma}\Gamma_{\text{sb}}^{-1} \right) \right) + \log \left( \frac{1}{\delta} \right)}{i\lambda_{\text{min}} \left( \Gamma_{\text{sb}} \right)}} \label{eqn:parameter-estimation-error-tail-bound-ineq-5}
    \end{align}
    
    Thus, we conclude that
    \begin{align}
        &P\left(  \norm{\hat{\theta}_i - \theta_*}_2 > d  \right) \\
        & \leq P\Bigg(  \norm{\hat{\theta}_i - \theta_*}_2 > \label{eqn:parameter-estimation-error-tail-bound-ineq-2} \\
        & \quad ( 90 \Sigma_W / p ) \sqrt{\frac{c_2(p) + \log \left( \det \left( \overline{\Gamma} \Gamma_{\text{sb}}^{-1} \right) \right) + \log \left( \frac{1}{\delta} \right) }{i\lambda_{\text{min}} \left( \Gamma_{\text{sb}} \right)}}  \Bigg) \\
        & \leq 3 \delta \label{eqn:parameter-estimation-error-tail-bound-ineq-3} \\
        &= 3i^{\frac{4}{3}} e^{-i \frac{\lambda_{\text{min}} \left( \Gamma_{\text{sb}} \right) d^2 p^2}{3 ( 90 \Sigma_W )^2}} \\
        & \quad \times c_1(q,\Gamma_{\text{sb}})^{\frac{2}{3}} e^{\frac{1}{3}\left( c_2(p) + \log \left( \det \left( \Gamma_{\text{sb}}^{-1} \right) \right) \right)} \label{eqn:parameter-estimation-error-tail-bound-ineq-4} \\
        &= i^{\frac{4}{3}} e^{- c_3(d,p,\Gamma_{\text{sb}}) i } c_4(p,q,\Gamma_{\text{sb}}).
    \end{align}
    where \eqref{eqn:parameter-estimation-error-tail-bound-ineq-2} follows from \eqref{eqn:parameter-estimation-error-tail-bound-ineq-5}, \eqref{eqn:parameter-estimation-error-tail-bound-ineq-3} follows from \eqref{eqn:parameter-estimation-error-tail-bound-ineq-6}, and \eqref{eqn:parameter-estimation-error-tail-bound-ineq-4} follows from \eqref{eqn:parameter-estimation-error-tail-bound-ineq-7}.
\end{proof}

\subsection{Analysis for Lemma \ref{lemma:msb-ce}}

We now provide the proof of Lemma \ref{lemma:msb-ce}. Following this, we provide the proofs for supporting results, namely Lemmas \ref{lemma:stability-a-not-1}, \ref{lemma:stability-reference-system} and \ref{lemma:bound-reference-error-after-enter-tube}. We then provide Lemma \ref{lemma:similar-controller-behaviour-after-time-T} and its proof.
\begin{proof}[Proof of Lemma \ref{lemma:msb-ce}]
    The proof proceeds by splitting the analysis into the case where 1) $a \in (-1,1)$, and 2) $a \in \{-1, 1 \}$.

    \textit{Case 1:} Consider the closed-loop process $(X_t)_{t\in \mathbb{N}_0}$ from \eqref{eqn:closed-loop-system-v2} with $a \in (-1,1)$. From Lemma \ref{lemma:stability-a-not-1}, we know that the process $(X_t)_{t \in \mathbb{N}_{0}}$ satisfies $\mathbb{E}\left[ X_t^2 \right] \leq x_0^2 + \frac{\beta(\lambda)}{1 - \lambda}$ for all $\lambda \in (a^2,1)$, with $\beta(\lambda)$ defined in Lemma \ref{lemma:stability-a-not-1}. The conclusion follows by choosing $\lambda \in (a^2,1)$, and setting $e = x_0^2 + \frac{\beta(\lambda)}{1-\lambda}$.
    
    \textit{Case 2}: Suppose $a \in \{-1,1\}$, and consider the process $(X_t)_{t \in \mathbb{N}_0}$. For all $t \in \mathbb{N}_0$, $\mathbb{E}[X_t^2]$ is upper bounded by 
    \begin{align}
        \mathbb{E}\left[ X_t^2 \right] &= \mathbb{E}\left[ \left( X_t - X_t^* + X_t^* \right)^2 \right] \\
        & \leq 2 \left( \mathbb{E}\left[ \left( X_t - X_t^* \right)^2  \right] + \mathbb{E}\left[ \left(X_t^*\right)^2 \right] \right), \label{eqn:decompose-mean-square}
    \end{align}
    where \eqref{eqn:decompose-mean-square} follows from $(a+b)^2 \leq 2(a^2 + b^2)$ for $a,b \in \mathbb{R}$, linearity of expectation, and the definition of $X_t^*$ \eqref{eqn:reference-system}. From Lemma \ref{lemma:stability-reference-system}, we know that there exists $e_1 > 0$ such that for all $t \in \mathbb{N}_0$, $\mathbb{E}\left[ (X_t^*)^2 \right] \leq e_1$. 
    
    Now, we aim to prove that there there exists $e_2 > 0$ such that for all $t \in \mathbb{N}_0$, $\mathbb{E}\left[ \left( X_t - X_t^* \right)^2 \right] \leq e_2$.
    
    Let $d^* := \min \left( m/2 , 1/2, \abs{b}/2 \right)$. Using the law of total expectation and the definition of $T_d$ \eqref{eqn:definition-Td}, for all $t \in \mathbb{N}_0$, $\mathbb{E}\left[ \left( X_t - X_t^* \right)^2 \right]$ can be upper bounded by
    \begin{align}
        &\mathbb{E}\left[ (X_t - X_t^*)^2 \right] 
        \\
        &\leq \sum_{k = 1}^{\infty} \mathbb{E} \left[ (X_t - X_t^*)^2 \mid  T_{d^*} = k \right] P \left(  T_{d^*} \geq k  \right). \label{eqn:total-expectation-deviation-bound}
    \end{align}
    
    From Lemma \ref{lemma:bound-reference-error-after-enter-tube}, we find that for all $k \in \mathbb{N}$ and $t \in \mathbb{N}_0$, on the event $\{ T_{d^*} = k \}$, $\abs{X_t - X_t^*} \leq \max \left( 2\abs{b}D (k+1), 2 \left( \frac{\abs{b}+d^*}{1-d^*} + 3\abs{b} \right) D \right)$. Using the monotonicity property of expectation, it follows that
    \begin{align}
        &\mathbb{E}\left[ \left( X_t - X_t^* \right)^2 \mid  T_{d^*} = k  \right] \\
        &\leq \mathbb{E}\left[ \max \left( 2\abs{b}D (k+1), 2 \left( \frac{\abs{b}+d^*}{1-d^*} + 3\abs{b} \right) D \right)^2 \mid  T_{d^*} = k  \right] \\
        &= \max \left( 4b^2D^2 (k+1)^2, 4\left( \frac{\abs{b}+d^*}{1-d^*} + 3 \abs{b} \right)^2D^2 \right) \label{eqn:conditional-expectation-deviation-bound}
    \end{align}
    Next, using the union bound, we find that for all $k \in \mathbb{N}$
    \begin{align}
        P(  T_{d^*} \geq k ) &= P \Big( \exists i \geq k-1 : \norm{\hat{\theta}_i - \theta_*}_2 > d^*  \Big) \\
        &\leq \sum_{i \geq k-1} P \Big( \norm{\hat{\theta}_i - \theta_*}_2 > d^* \Big). \label{eqn:td-tail-bound}
    \end{align}
    Combining \eqref{eqn:total-expectation-deviation-bound}, \eqref{eqn:conditional-expectation-deviation-bound} and \eqref{eqn:td-tail-bound} we find
    \begin{equation}
        \begin{aligned}
        &\mathbb{E}\left[ (X_t-X_t^*)^2 \right] \\
        & \leq \sum_{k = 1}^{\infty} \max \Big( 4b^2D^2 (k+1)^2, 4\Big( \frac{\abs{b}+d}{1-d} + 3 \abs{b} \Big)^2D^2 \Big) \\
        & \quad \times \sum_{i \geq k-1} P\Big(\norm{\hat{\theta}_i - \theta_*}_2 > d^*\Big) =: e_2
        \end{aligned} \label{eqn:e-dev-upper-bound} \\
    \end{equation}
    for $k \in \mathbb{N}$, where we introduce $e_2$ to denote the infinite sum which uniformly bounds $\mathbb{E}\left[ (X_t-X_t^*)^2 \right]$ for all $t \in \mathbb{N}_0$. Now, let $N = \left \lceil \frac{\frac{\abs{b}+d^*}{1-d^*} + 3 \abs{b}}{\abs{b}} \right \rceil - 1$. We find that
    \begin{align}
        &\sum_{k = N}^{\infty} \max \Big( 4b^2D^2 (k + 1)^2, 4\Big( \frac{\abs{b}+d^*}{1-d^*} + 3 \abs{b} \Big)^2D^2 \Big) \\
        & \quad \times \sum_{i \geq k-1} P\Big(\norm{\hat{\theta}_i - \theta_*}_2 > d^*\Big) \\
        & = 4 b^2 D^2 \sum_{k = N}^{\infty} (k+1)^2 \sum_{i \geq k-1} P\Big(\norm{\hat{\theta}_i - \theta_*}_2 > d^*\Big) \label{eqn:msb-ce-ineq-1} \\
        & \leq 4 b^2 D^2  \sum_{k = 1}^{\infty} (k+1)^2 \sum_{i \geq k-1} P\Big(\norm{\hat{\theta}_i - \theta_*}_2 > d^*\Big) \\
        & \leq 8 b^2 D^2  \sum_{k = 1}^{\infty} k^2 \sum_{i \geq k-1} P\Big(\norm{\hat{\theta}_i - \theta_*}_2 > d^*\Big) \\
        & < \infty \label{eqn:msb-ce-ineq-2}
    \end{align}
    where \eqref{eqn:msb-ce-ineq-1} follows since $\max \big( 4b^2D^2 (k+1)^2, 4\big( \frac{\abs{b}+d^*}{1-d^*} + 3 \abs{b} \big)^2D^2 \big) = 4 b^2 D^2 (k + 1)^2$ for $k \geq N$, and \eqref{eqn:msb-ce-ineq-2} follows from the the assumption that  $\sum_{k=1}^{\infty}k^2 \sum_{i \geq k-1} P\left( \norm{\hat{\theta}_i - \theta_*}_2 >d  \right) < \infty$ for all $d \in (0,m)$ in the premise. Moreover, we have that
    \begin{align}
        &\sum_{k = 1}^{N-1} \max \Big( 4b^2D^2 (k+1)^2, 4\Big( \frac{\abs{b}+d^*}{1-d^*} + 3 \abs{b} \Big)^2D^2 \Big) \\
        &\quad \times \sum_{i \geq k-1} P\Big(\norm{\hat{\theta}_i - \theta_*}_2 > d^*\Big) \\
        & = 4\Big( \frac{\abs{b}+d^*}{1-d^*} + 3 \abs{b} \Big)^2D^2 \\
        & \quad \times \sum_{k = 1}^{N-1} \sum_{i \geq k-1} P\Big(\norm{\hat{\theta}_i - \theta_*}_2 > d^*\Big) \label{eqn:msb-ce-ineq-3} \\
        & \leq 4\Big( \frac{\abs{b}+d^*}{1-d^*} + 3 \abs{b} \Big)^2D^2  \\
        & \quad \times \sum_{k = 1}^{\infty} k^2 \sum_{i \geq k-1} P\Big(\norm{\hat{\theta}_i - \theta_*}_2 > d^*\Big) \\
        & < \infty \label{eqn:msb-ce-ineq-4}
    \end{align}
    where \eqref{eqn:msb-ce-ineq-3} follows since $\max \big( 4b^2D^2 (k+1)^2, 4\big( \frac{\abs{b}+d^*}{1-d^*} + 3 \abs{b} \big)^2D^2 \big) = 4\big( \frac{\abs{b}+d^*}{1-d^*} + 3 \abs{b} \big)^2D^2$ for $k < N$, and \eqref{eqn:msb-ce-ineq-4} follows from the premise.
    From \eqref{eqn:msb-ce-ineq-4}, \eqref{eqn:msb-ce-ineq-2} and \eqref{eqn:e-dev-upper-bound}, it follows that $e_2 < \infty$. Our conclusion follows by setting $e = 2\left( e_1 + e_2 \right)$.
\end{proof}

\begin{proof}[Proof of Lemma \ref{lemma:stability-a-not-1}]
    Suppose $t \in \mathbb{N}_0$. Recall the closed-loop system from \eqref{eqn:closed-loop-system-v2}. Squaring this, we obtain,
    \begin{align}
        X_{t+1}^2 &= \left( \abs{a}\abs{X_t} + \abs{b}U_t + \abs{W_t} \right)^2 \\
        &\leq \left( \abs{a}\abs{X_t} + \abs{b}U_{\text{max}} + \abs{W_t} \right)^2 \label{eqn:stability-a-not-1-umax-bound-step} \\
        &= a^2 X_t^2 + 2 \abs{a}\abs{X_t} C_t + C_t^2 \label{eqn:stability-a-not-1-bounding-square-state}
    \end{align}
    where \eqref{eqn:stability-a-not-1-umax-bound-step} holds since G1 is satisfied by construction, and \eqref{eqn:stability-a-not-1-bounding-square-state} holds by the definition $C_{t} := \abs{b}U_{\text{max}} + \abs{W_{t}}$. Note that the first and second moments of $C_{t}$ satisfy $\mathbb{E}\left[ C_{t} \right] = D_1$ and $\mathbb{E}\left[ C_{t}^2 \right] = D_2$.
    
    Now define $K(\lambda) := \left\{ x \in \mathbb{R} : |x| \leq E(\lambda) \right\}$ for all $\lambda \in (a^2,1)$. Suppose $\lambda \in (a^2,1)$. On the event $\{ X_{t} \not \in K(\lambda) \}$, we have
    \begin{align}
        \mathbb{E}\left[ X_{t+1}^2 \mid X_{t} \right] &\leq  {a}^2 {X_{t}}^2 + 2 \abs{a} \abs{X_{t}} D_1 + D_2 \label{eqn:stability-a-not-1-bounding-conditional-outside-K-middle-step} \\
        &\leq \lambda X_{t}^2 \label{eqn:stability-a-not-1-bounding-conditional-outside-K-last-step}
    \end{align}
    where \eqref{eqn:stability-a-not-1-bounding-conditional-outside-K-middle-step} follows from \eqref{eqn:stability-a-not-1-bounding-square-state} since $C_t$ is independent of $X_t$, and \eqref{eqn:stability-a-not-1-bounding-conditional-outside-K-last-step} follows from the fact that on the event $\{ X_{t} \not \in K(\lambda) \}$, $\abs{X_{t}} > E(\lambda)$ holds (by definition), as well as the fact that $|x| > E(\lambda) \implies \lambda x^2 \geq \abs{a}^2 |x|^2 + 2 |a| |x| D_1 + D_2$ (seen by applying the quadratic formula to solve for the set of $\abs{x}$ such that $(\lambda-\abs{a}^2)\abs{x}^2 - 2\abs{a}\abs{x}D_1-D_2 \geq 0$). On the event $\{X_{t} \in K(\lambda) \}$, we have,
    \begin{align}
        \mathbb{E}\left[ X_{t+1}^2 | X_{t} \right] &\leq {a}^2 {X_{t}}^2 + 2 \abs{a} \abs{X_{t}} D_1 + D_2 \\
        & \leq {a}^2 E(\lambda)^2 + 2 \abs{a} E D_1 + D_2  = \beta(\lambda) \label{eqn:stability-a-not-1-ineq-1}
    \end{align}
    where \eqref{eqn:stability-a-not-1-ineq-1} follows since on the event $\{X_{t} \in K(\lambda) \}, \abs{X_{t}} \leq E(\lambda)$ holds (from the definition of $K(\lambda)$).
    Finally, we find
    \begin{align}
        \mathbb{E}\big[ X_t^2 \big] &= \mathbb{E}\big[ \mathbb{E}\big[X_t^2|X_{t-1}\big] \big] \allowdisplaybreaks \\
        &= \mathbb{E}\big[ \mathbb{E}\big[X_t^2|X_{t-1}\big] \bm{1}_{ \{X_{t-1} \not \in K(\lambda) \}} \big] \allowdisplaybreaks \\
        & \quad + \mathbb{E}\big[ \mathbb{E}\big[X_t^2|X_{t-1}\big] \bm{1}_{\{ X_{t-1} \in K(\lambda) \}} \big] \allowdisplaybreaks \\
        &\leq \mathbb{E}\big[ \lambda X_{t-1}^2 \bm{1}_{\{ X_{t-1} \not \in K(\lambda) \}}\big] \label{eqn:geometric-stability-conditions-step} \allowdisplaybreaks \\ 
        & \quad + \mathbb{E}\big[\beta \bm{1}_{ \{ X_{t-1} \in K(\lambda) \}}  \big] \allowdisplaybreaks \\
        &\leq \lambda \mathbb{E}\big[ X_{t-1}^2 \big] + \beta(\lambda) \label{eqn:geometric-stability-pre-iterative-step} \allowdisplaybreaks \\
        & \leq \lambda^t \mathbb{E}\big[ X_0^2 \big] + \beta(\lambda) \sum_{k=0}^{t-1}\lambda^{t-1-k} \label{eqn:geometric-stability-post-iterative-step} \allowdisplaybreaks \\
        & \leq x_0^2 + \frac{\beta(\lambda)}{1 - \lambda} \label{eqn:geometric-stability-infinite-sum-sequence-step}
    \end{align}
    for all $t \in \mathbb{N}_0$, where \eqref{eqn:geometric-stability-conditions-step} follows from conditions \eqref{eqn:stability-a-not-1-bounding-conditional-outside-K-last-step} and \eqref{eqn:stability-a-not-1-ineq-1}, \eqref{eqn:geometric-stability-post-iterative-step} follows by iteratively applying \eqref{eqn:geometric-stability-pre-iterative-step}, and \eqref{eqn:geometric-stability-infinite-sum-sequence-step} follows from $X_0^2=x_0^2$ and the infinite sum of a geometric sequence.
\end{proof}

\begin{proof}[Proof of Lemma \ref{lemma:stability-reference-system}] Define $Y_t^* := a^t X_t^*$, $t \in \mathbb{N}_0$. Then, $\mathbb{E}\left[ \left(X_t^*\right)^2 \right]$ can be equivalently rewritten as follows for all $t \in \mathbb{N}_0$:
\begin{align}
    \mathbb{E}\big[\big(X_t^*\big)^2\big] &= \mathbb{E}\big[\big(Y_t^*\big)^2\big] \\
    &= \mathbb{E}\big[ \big(\big(Y_t^*\big)^+\big)^2 \big]+ \mathbb{E}\big[\big( (-Y_t^*)^+ \big)^2 \big] \label{eqn:stability-reference-system-middle-disappear-step}
\end{align}
where \eqref{eqn:stability-reference-system-middle-disappear-step} follows from the properties of $(\cdot)^+$ and linearity of expectation.
We will prove that there exists $e_1 >0$ such that for all $t \in \mathbb{N}_0$, $\mathbb{E}\left[ \left(\left(Y_t^*\right)^+\right)^2 \right] \leq e_1$. This will be accomplished by showing that $(Y^*_t)_{t \in \mathbb{N}_0}$ satisfies all of the conditions in Proposition \ref{prop:constant-drift-conditions}. In particular, the conditions are satisfied with $\gamma = \abs{b}D$, $J = \abs{b}D + |x_0|$, and $M = 8 \left( b^4 U_{\text{max}}^4 + S_4 \right)$. Firstly, note that the condition $Y^*_0 \leq J$ is satisfied since $Y^*_0 = a^0 X_0 \leq \abs{b} D + |x_0|$.

Next, we verify condition \eqref{eqn:msb-condition-1}. Note that
\begin{align}
    Y^*_{t+1} &= a^{t+1}X_{t+1}^* \\
    &= a^2 Y^*_t + a^{t+1} b \sigma_D \left( - \frac{a}{b} X_t^* \right) \\
    & \quad + a^{t+1} \left( b V_t + W_t \right) \label{eqn:stability-reference-system-y-definition-step-2} \\
    &= Y^*_t - \sigma_{\abs{b}D}\left( Y^*_t \right) + a^{t+1}\left( bV_t + W_t \right), \label{eqn:stability-reference-system-y-evolution}
\end{align}
where \eqref{eqn:stability-reference-system-y-definition-step-2} holds from the closed-loop system \eqref{eqn:closed-loop-system-v2} and the definition of $Y_t^*$, and \eqref{eqn:stability-reference-system-y-evolution} is due to the following equality:
\begin{align}
    &a^{t+1} b \sigma_D \left( - \frac{a}{b} X_t \right)  \\
    &= \begin{cases}
        b  a^{t-1} \left( - \frac{1}{b} a^{-t + 1} Y^*_t \right), \abs{- \frac{1}{b} a^{-t + 1} Y^*_t} \leq D \\
        b  a^{t-1} \frac{\left( - \frac{1}{b} a^{-t + 1} Y^*_t \right)}{\abs{ - \frac{1}{b} a^{-t + 1} Y^*_t }} D, \abs{- \frac{1}{b} a^{-t + 1} Y^*_t} > D 
    \end{cases} \label{eqn:simplify-saturation_y-sigma-definition-step-1} \\
    &= \begin{cases}
            \left( - Y^*_t \right), \abs{Y^*_t} \leq D \abs{b} \\
            \frac{\left( - Y^*_t \right)}{\abs{- Y^*_t}} D \abs{b}, \abs{Y^*_t} > D \abs{b}
    \end{cases} = \sigma_{\abs{b}D}\left( -Y^*_t \right). \label{eqn:simplify-saturation_y-sigma-definition-step-2}
\end{align}
where both \eqref{eqn:simplify-saturation_y-sigma-definition-step-1} and \eqref{eqn:simplify-saturation_y-sigma-definition-step-2} follow from the definition of $\sigma_D(\cdot)$. 
Let $\mathcal{F}_t$ be the natural filtration of the process $(Y^*_t)_{t \in \mathbb{N}_0}$. For all $t \in \mathbb{N}_0$, on the event $\{ Y^*_t > \abs{b}D + |x_0| \}$, we have
\begin{align}
    &\mathbb{E}\left[ Y^*_{t+1} - Y^*_t \mid \mathcal{F}_t \right]  \\
    &= \mathbb{E}\left[ -\sigma_{\abs{b} D} \left( Y^*_t \right) + a^{t+1}\left( b V_t + W_t \right) \mid \mathcal{F}_t \right] \label{eqn:stability-reference-expected-difference-y-definition-step} \\
    &= -\abs{b}D \label{eqn:stability-reference-expected-difference-y-saturated-step} = -\gamma,
\end{align}
where \eqref{eqn:stability-reference-expected-difference-y-definition-step} holds due to \eqref{eqn:stability-reference-system-y-evolution}, and \eqref{eqn:stability-reference-expected-difference-y-saturated-step} holds since $\sigma_{\abs{b}D}(Y_t^*) = \abs{b}D$ when $Y_t^* > \abs{b}D$, and $\mathbb{E}\left[ bV_t + W_t \mid \mathcal{F}_t \right] = 0$.
Thus, condition \eqref{eqn:msb-condition-1} has been verified. 

We now verify condition \eqref{eqn:msb-condition-2} as follows:
\begin{align}
    &\mathbb{E}\left[ \abs{Y^*_{t+1} - Y^*_t}^4 \mid Y^*_0, \hdots, Y^*_t \right] \\
    &\leq \mathbb{E}\left[ \left(\abs{b}(D+C) + \abs{W_t}\right)^4 \mid Y^*_0, \hdots, Y^*_t \right] \label{eqn:stability-reference-absolute-difference-2} \\
    & \leq 8 \left( b^4 U_{\text{max}}^4 + \mathbb{E}\left[ \abs{W_t}^4 \right] \right) \label{eqn:stability-reference-absolute-difference-3} = M,
\end{align}
where \eqref{eqn:stability-reference-absolute-difference-2} follows from \eqref{eqn:stability-reference-system-y-evolution} and the definition of $\sigma_{\abs{b}D}(\cdot)$, and \eqref{eqn:stability-reference-absolute-difference-3} follows from $(a+b)^4 \leq 8(a^4 + b^4)$ for $a,b \in \mathbb{R}$ and linearity of conditional expectation.

Therefore, since $(Y^*_t)_{t \in \mathbb{N}_0}$ satisfies the conditions in Proposition \ref{prop:constant-drift-conditions}, we find that there exists $e_1 >0$ such that for all $t \in \mathbb{N}_0$, $\mathbb{E}\left[ \left(\left(Y^*_t\right)^+\right)^2 \right] \leq e_1$.

Following an analogous method, we are also able to establish that the process $(-Y^*_t)_{t \in \mathbb{N}_0}$ satisfies the conditions in Proposition \ref{prop:constant-drift-conditions}, so there exists $e_2 >0$ such that for all $t \in \mathbb{N}_0$, $\mathbb{E}\left[ \left(\left(-Y_t^*\right)^+\right)^2 \right] \leq e_2$. Setting $c = e_1 + e_2$, it follows that $\mathbb{E}\left[ \left( Y^*_t \right)^2 \right] \leq c$, and therefore $\mathbb{E}\left[ \left( X^*_t  \right)^2 \right] \leq c $.
\end{proof}

\begin{proof}[Proof of Lemma \ref{lemma:bound-reference-error-after-enter-tube}]
    For all $t \in \mathbb{N}_0$, the error $X_{t+1}-X_{t+1}^*$ evolves as
    \begin{align}
        &X_{t+1}-X_{t+1}^* \\
        &=  \Big(aX_t + b \Big(\sigma_D\Big(G_tX_t\Big) + V_{t} \Big) + W_t \Big) \label{eqn:bound-reference-error-after-enter-tube-eq-1} \\
        & \quad - \Big( aX_t^* + b \Big(\sigma_D\Big(-\frac{a}{b}X_t^*\Big) + V_{t} \Big) + W_t \Big) \\
        &= a ( X_t - X_t^* ) \\
        & \quad + b \Big(\sigma_D\Big(G_tX_t\Big) - \sigma_D\Big(-\frac{a}{b}X_t^*\Big) \Big) \label{eqn:bound-reference-error-after-enter-tube-eq-2}
    \end{align}
    where \eqref{eqn:bound-reference-error-after-enter-tube-eq-1} follows from \eqref{eqn:closed-loop-system-v2}.
    Taking the absolute value, we find $\abs{X_{t+1}-X_{t+1}^*}$ is upper bounded in terms of $\abs{X_{t}-X_{t}^*}$ as follows
    \begin{align}
        \abs{X_{t+1}-X_{t+1}^*} &= \Big|a \left( X_t - X_t^* \right) \\
        & \quad + b\Big( \sigma_D\Big(G_tX_t\Big) - \sigma_D\Big(-\frac{a}{b}X_t^*\Big) \Big) \Big| \\
        & \leq \abs{X_t - X_t^*} \\
        & \quad + \abs{b} \Big|\sigma_D\Big(G_tX_t\Big) - \sigma_D\Big(-\frac{a}{b}X_t^*\Big)\Big| \\
        & \leq \abs{X_t - X_t^*} + \abs{b}2D  \label{eqn:bound-reference-error-after-enter-tube-ineq-1}
    \end{align}
    By iteratively applying \eqref{eqn:bound-reference-error-after-enter-tube-ineq-1}, the following then holds for $k \in \mathbb{N}_0$, and $t \leq k+1$:
    \begin{align}
        \abs{X_t - X_t^*} \leq t \abs{b} 2D \leq (k+1) \abs{b} 2D
    \end{align}
    
    We now move onto the case where $t \geq k+1$. Firstly, define the processes $(Y_t)_{t\in \mathbb{N}_0}$ and $(Y_t^*)_{t \in \mathbb{N}_0}$ so $Y_t := a^t X_t$ and $Y_t^*:= a^t X_t^*$. Their difference $Y_{t+1}  - Y_{t+1}^*$ can be written as
    \begin{align}
        &Y_{t+1}  - Y_{t+1}^* \\
        &= a^{t+1}( X_{t+1} - X_{t+1}^*) \\
        &= a^{t+1} \Big( a ( X_t - X_t^* ) \\
        & \quad + b\Big( \sigma_D\Big(G_tX_t\Big) - \sigma_D\Big(-\frac{a}{b}X_t^*\Big) \Big) \Big) \label{eqn:bound-reference-error-after-enter-tube-eq-3} \\
        &= Y_t - Y_t^* \\
        & \quad +  a^{t+1} b \Big( \sigma_D\Big(G_tX_t\Big) - \sigma_D\Big(-\frac{a}{b}X_t^*\Big)  \Big) \label{eqn:bound-reference-error-after-enter-tube-eq-4}
    \end{align}
    where \eqref{eqn:bound-reference-error-after-enter-tube-eq-3} follows from \eqref{eqn:bound-reference-error-after-enter-tube-eq-2}.

    Next, let $\Omega$ denote the underlying sample space. Suppose $d \in (0,\min(1,\abs{b}))$. Let $H=(\abs{b}+d)/(1-d)$, and let $E_1 = \{ \abs{Y_t} > H D \}$, $E_2 = \{ \abs{Y_t^*} > H D \}$, $E_3 = \{ \abs{Y_t} > (H + 2\abs{b}) D \}$, $E_4 = \{ \abs{Y_t^*} > (H + 2 \abs{b} ) D \}$. Moreover, let $A_1 = \{ \abs{Y_t}>HD , \abs{Y_t^*} > H D \}$, $A_2 = \{ \abs{Y_t}>(H + 2 \abs{b})D, \abs{Y_t^*} \leq H D \}$, $A_3 = \{ \abs{Y_t} \leq HD , \abs{Y_t^*} > (H + 2\abs{b})D \}$, and $A_4 = \{ \abs{Y_t} \leq (H+2\abs{b})D , \abs{Y_t^*} \leq (H+2\abs{b}) D \}$. The sample space $\Omega$ can be equivalently written as
    \begin{align}
        \Omega &= (E_1 \cap E_2) \cup (E_1 \cap E_2^C) \cup (E_1^C \cap E_2^C) \cup (E_1^C \cap E_2) \\
        &= (E_1 \cap E_2) \cup (E_1 \cap E_2^C \cap E_3) \cup (E_1 \cap E_2^C \cap E_3^C) \cup \\
        & \quad  (E_1^C \cap E_2^C) \cup (E_1^C \cap E_2 \cap E_4) \cup (E_1^C \cap E_2 \cap E_4^C) \\
        &= (E_1 \cap E_2) \cup (E_3 \cap E_2^C) \cup (E_3^C \cap E_2^C) \cup \\
        & \quad (E_1^C \cap E_2^C) \cup (E_1^C \cap E_4) \cup (E_1^C \cap E_4^C) \label{eqn:bound-reference-error-set-ineq-1} \\
        &= A_1 \cup A_2 \cup A_3 \cup A_4, \label{eqn:bound-reference-error-set-ineq-2}
    \end{align}
    with \eqref{eqn:bound-reference-error-set-ineq-2} holding since $A_1 = E_1 \cap E_2$, $A_2 = E_3 \cap E_2^C$, $A_3 = E_1^C \cap E_4$, and $(E_3^C \cap E_2^C) \cup (E_1^C \cap E_2^C) \cup (E_1^C \cap E_4^C) \subseteq A_4 \subseteq \Omega$. Since $\abs{Y_{t+1}-Y_{t+1}^*}$ takes values in $\mathbb{R}_{\geq 0}$, making use of the properties of the indicator function, we find
    \begin{align}
        &\abs{Y_{t+1} - Y_{t+1}^*} = \max \big(\abs{Y_{t+1}  - Y_{t+1}^* } \bm{1}_{A_1}, \\
        & \quad \abs{Y_{t+1}  - Y_{t+1}^* }\bm{1}_{A_2}, \abs{Y_{t+1}  - Y_{t+1}^* }\bm{1}_{A_3}, \\
        & \quad \abs{Y_{t+1}  - Y_{t+1}^* }\bm{1}_{A_4} \big). \label{eqn:reference-error-case-bound}
    \end{align}
    We now prove upper bounds for $\abs{Y_{t+1}-Y_{t+1}^*}$ on the event $\{ T_d = k \} \cap A_i$ for $i \in \{ 1, \hdots, 4 \}$ and $t \geq k + 1$.
        
    \textit{Case 1:} Consider the event $\{T_d = k \} \cap A_1$, and suppose $t \geq k + 1$. Since $\abs{Y_t} > H D$, $\abs{Y_t^*} > H D$ on $A_1$, we have
    \begin{align}
        &Y_{t+1}  - Y_{t+1}^* \allowdisplaybreaks \\
        &= Y_t - Y_t^* +  a^{t+1} b \bigg( \sigma_D\bigg(-\frac{\hat{a}_{t-1}}{\hat{b}_{t-1}}X_t\bigg) \\
        & \quad  - \sigma_D\bigg(-\frac{a}{b}X_t^*\bigg)  \bigg) \allowdisplaybreaks \label{eqn:bound-reference-error-after-enter-tube-eq-5} \\
        &= Y_t - Y_t^*  +  a^{t+1} b \bigg( -\frac{\hat{a}_{t-1}/\hat{b}_{t-1}}{ |{\hat{a}_{t-1}/\hat{b}_{t-1}}|}\frac{X_t}{\abs{X_t}}D \\
        & \quad  + \frac{a/b}{\abs{a/b}}\frac{X_t^*}{\abs{X_t^*}}D \bigg) \label{eqn:bound-reference-error-after-enter-tube-eq-6}  \allowdisplaybreaks \\
        &= Y_t - Y_t^* \\
        & \quad +  a^{t+1} b \left( -\frac{a/b}{\abs{a/b}}\frac{X_t}{\abs{X_t}}D + \frac{a/b}{\abs{a/b}}\frac{X_t^*}{\abs{X_t^*}}D \right) \label{eqn:bound-reference-error-after-enter-tube-eq-12} \allowdisplaybreaks \\
        &= Y_t - Y_t^* -  \abs{b} \left(\frac{Y_t}{\abs{Y_t}} - \frac{Y_t^*}{\abs{Y_t^*}} \right) D
    \end{align}
    where \eqref{eqn:bound-reference-error-after-enter-tube-eq-5} follows from \eqref{eqn:bound-reference-error-after-enter-tube-eq-4} and \eqref{eqn:control-policy-v2}, \eqref{eqn:bound-reference-error-after-enter-tube-eq-6} follows from the definition of $\sigma_D(\cdot)$, and \eqref{eqn:bound-reference-error-after-enter-tube-eq-12} follows from Lemma \ref{lemma:similar-controller-behaviour-after-time-T}.
    
    If $Y_t > H D$ and $Y_t^* > H D$, or $Y_t < - H D$ and $Y_t^* < - H D$, then $\frac{Y_t}{\abs{Y_t}} - \frac{Y_t^*}{\abs{Y_t^*}} = 0$ and so $Y_{t+1}  - Y_{t+1}^* = Y_t - Y_t^*$.
    
    If $Y_t > H D$ and $Y_t^* < -H D$, then $\frac{Y_t}{\abs{Y_t}} - \frac{Y_t^*}{\abs{Y_t^*}} = 2$, and so $\abs{Y_{t+1}-Y_{t+1}^*} = Y_{t+1} - Y_{t+1}^* = Y_t - Y_t^* - \abs{b}2D \leq Y_t - Y_t^* = \abs{Y_t - Y_t^*} $.
    
    If $Y_t < -H D$ and $Y_t^* > H D$, then $\frac{Y_t}{\abs{Y_t}} - \frac{Y_t^*}{\abs{Y_t^*}} = -2$, and so $\abs{Y_{t+1}-Y_{t+1}^*} = \abs{Y_t - Y_t^* + \abs{b}2D} = (-1) \left( Y_t - Y_t^* + 2 \abs{b} D \right) = (-1) \left( Y_t - Y_t^* \right) - 2 \abs{b} D = \abs{Y_t - Y_t^*} - 2\abs{b} D \leq \abs{Y_t - Y_t^*}$.
    
    Thus, it follows that on the event $\{T_d = k \} \cap A_1$, $\abs{Y_{t+1} - Y_{t+1}^*} \leq \abs{Y_t - Y_t^*}$ holds.
    
    \textit{Case 2:} Consider the event $\{T_d = k \} \cap A_2$, and suppose $t \geq k + 1$. Since $\abs{Y_t} > \left( H + 2 \abs{b} \right) D$ and $\abs{Y_t^*} \leq H D$ on $A_2$, we have
    \begin{align}
        &Y_{t+1}-Y_{t+1}^* \\
        &= Y_t - Y_t^* +  a^{t+1} b \Big( \sigma_D\Big(-\frac{\hat{a}_{t-1}}{\hat{b}_{t-1}}X_t\Big) \\
        & \quad  - \sigma_D\Big(-\frac{a}{b}X_t^*\Big)  \Big) \label{eqn:bound-reference-error-after-enter-tube-eq-7} \\
        &= Y_t - Y_t^* + a^{t+1}b\Big( - \frac{\hat{a}_{t-1}/\hat{b}_{t-1}}{\big|{\hat{a}_{t-1}/\hat{b}_{t-1}}\big|} \frac{X_t}{\abs{X_t}}D \\
        & \quad  - \sigma_D\Big( \frac{-a}{b} X_t^* \Big) \Big) \label{eqn:bound-reference-error-after-enter-tube-eq-8} \\
        &= Y_t - Y_t^* \\
        & \quad + a^{t+1}b\Big( - \frac{a/b}{\abs{a/b}} \frac{X_t}{\abs{X_t}}D - \sigma_D\Big( \frac{-a}{b} X_t^* \Big) \Big) \label{eqn:bound-reference-error-after-enter-tube-eq-9} \\
        &= Y_t - Y_t^* - \frac{Y_t}{\abs{Y_t}} \abs{b} D + a^{t+1}b \Big( - \sigma_D\Big( \frac{-a}{b} X_t^* \Big) \Big)
    \end{align}
    where \eqref{eqn:bound-reference-error-after-enter-tube-eq-7} follows from \eqref{eqn:bound-reference-error-after-enter-tube-eq-4} and \eqref{eqn:control-policy-v2}, \eqref{eqn:bound-reference-error-after-enter-tube-eq-8} follows from the definition of $\sigma_D(\cdot)$ and \eqref{eqn:bound-reference-error-after-enter-tube-eq-9} follows from Lemma \ref{lemma:similar-controller-behaviour-after-time-T}.
    
    Consider when $Y_t > (H + 2 \abs{b}) D$. Then, $Y_{t+1} - Y_{t+1}^* = Y_t - Y_t^* - \abs{b}D + a^{t+1}b \left( - \sigma_D\left( \frac{-a}{b} X_t^* \right) \right)$. Since $H D \geq Y_t^*$, then $Y_t - Y_t^* \geq (H + 2 \abs{b})D - H D = 2\abs{b} D$. Moreover, since $\abs{a^{t+1}b \left( - \sigma_D\left( \frac{-a}{b} X_t^* \right) \right)} \leq \abs{b}D$, then $Y_t - Y_t^* - \abs{b}D + a^{t+1}b \left( - \sigma_D\left( \frac{-a}{b} X_t^* \right) \right) \geq 0 $. Therefore, $\abs{Y_{t+1} - Y_{t+1}^*}=\abs{Y_t - Y_t^* - \abs{b}D + a^{t+1}b \left( - \sigma_D\left( \frac{-a}{b} X_t^* \right) \right)} = Y_t - Y_t^* - \abs{b}D + a^{t+1}b \left( - \sigma_D\left( \frac{-a}{b} X_t^* \right) \right) \leq Y_t - Y_t^* = \abs{Y_t - Y_t^*}$.
    
    Now, consider when $Y_t < - (H + 2 \abs{b}) D$. Then, $Y_{t+1} - Y_{t+1}^* = Y_t - Y_t^* + \abs{b}D + a^{t+1}b \left( - \sigma_D\left( \frac{-a}{b} X_t^* \right) \right)$. Since $Y_t^* \geq -H D$, then $Y_t - Y_t^* \leq -(H + 2 \abs{b})D + H D = -2\abs{b} D$. Moreover, since $\abs{a^{t+1}b \left( - \sigma_D\left( \frac{-a}{b} X_t^* \right) \right)} \leq \abs{b}D$, then $Y_t - Y_t^* + \abs{b}D + a^{t+1}b \left( - \sigma_D\left( \frac{-a}{b} X_t^* \right) \right) \leq 0 $. Therefore, $\abs{Y_{t+1} - Y_{t+1}^*}=\abs{Y_t - Y_t^* + \abs{b}D + a^{t+1}b \left( - \sigma_D\left( \frac{-a}{b} X_t^* \right) \right)} = - \left(Y_t - Y_t^* + \abs{b}D + a^{t+1}b \left( - \sigma_D\left( \frac{-a}{b} X_t^* \right) \right) \right) = -(Y_t - Y_t^*) - \abs{b}D - a^{t+1}b \left( - \sigma_D\left( \frac{-a}{b} X_t^* \right) \right) \leq -(Y_t - Y_t^*) = \abs{Y_t - Y_t^*}$.

    Thus, it follows that on the event $\{T_d = k \} \cap A_2$, $\abs{Y_{t+1} - Y_{t+1}^*} \leq \abs{Y_t - Y_t^*}$ holds.
        
    \textit{Case 3:} Consider the event $\{T_d = k \} \cap A_3$, and suppose $t \geq k + 1$. Case 3 follows similarly to Case 2, but it does not require Lemma \ref{lemma:similar-controller-behaviour-after-time-T}. In particular, \eqref{eqn:bound-reference-error-after-enter-tube-eq-4}, \eqref{eqn:control-policy-v2}, and the definition of $\sigma_D(\cdot)$, are first used to establish that $Y_{t+1}-Y_{t+1}^*= Y_t - Y_t^* + a^{t+1}b \left( \sigma_D\left( -\frac{\hat{a}_{t-1}}{\hat{b}_{t-1}} X_t \right) \right) + \frac{Y_t^*}{\abs{Y_t^*}} \abs{b} D$. Following this, we then prove that $\abs{Y_t - Y_t^* + a^{t+1}b \left( \sigma_D\left( -\frac{\hat{a}_{t-1}}{\hat{b}_{t-1}} X_t \right) \right) + \frac{Y_t^*}{\abs{Y_t^*}} \abs{b} D} \leq \abs{Y_t - Y_t^*}$ for both $Y_t^* > (H+2\abs{b})D$ and $Y_t^* < -(H + 2\abs{b})D$ following analogous steps to Case 2, thereby establishing $\abs{Y_{t+1} - Y_{t+1}^*} \leq \abs{Y_t - Y_t^*}$.
    
    \textit{Case 4:} Consider the event $\{T_d = k \} \cap A_4$, and suppose $t \geq k + 1$. The following holds:    
    \begin{align}
        &\abs{Y_{t+1} - Y_{t+1}^*} \\
        &= \abs{Y_t - Y_t^* +  a^{t+1} b \left( \sigma_D\left(-\frac{\hat{a}_{t-1}}{\hat{b}_{t-1}}X_t\right) - \sigma_D\left(-\frac{a}{b}X_t^*\right)  \right)} \\
        & \leq \abs{Y_t - Y_t^*} + \abs{ b} \abs{ \left( \sigma_D\left(-\frac{\hat{a}_{t-1}}{\hat{b}_{t-1}}X_t\right) - \sigma_D\left(-\frac{a}{b}X_t^*\right)  \right)} \\
        & \leq 2 (H+2\abs{b}) D + \abs{ b} 2D \label{eqn:bound-reference-error-after-enter-tube-eq-13}\\
        & = 2 \left( H + 3\abs{b} \right) D
    \end{align}
    where \eqref{eqn:bound-reference-error-after-enter-tube-eq-13} follows from $\abs{Y_t} \leq (H+2\abs{b}) D$, $\abs{Y_t^*} \leq (H+2\abs{b}) D$ and the definition of $\sigma_D(\cdot)$.
    
    Combining the upper bounds from Cases 1-4 with \eqref{eqn:reference-error-case-bound}, it follows that on the event $\{T_d = k \}$, the following holds for all $t \geq k + 1$:
    \begin{equation}
        \abs{Y_{t+1} - Y_{t+1}^*} \leq \max ( \abs{Y_t - Y_t^*}, 2 \left( H + 3\abs{b} \right) D ). \label{eqn:bound-reference-error-after-enter-tube-evolution-bound}
    \end{equation}
    Iteratively applying \eqref{eqn:bound-reference-error-after-enter-tube-evolution-bound}, we find that $\abs{Y_t - Y_t^*} \leq \max \left( \abs{Y_{k+1} - Y_{k+1}^*}, 2 \left( H + 3\abs{b} \right) D  \right)$ for all $t \geq k + 1$.
    We can relate this back to $\abs{X_{t} - X_{t}^*}$ since $\abs{X_{t} - X_{t}^*} = \abs{a^{-t}Y_{t} - a^{-t}Y_{t}^*} = \abs{a^{-t}} \abs{Y_{t} - Y_{t}^*} = \abs{Y_{t} - Y_{t}^*}$ for all $t \in \mathbb{N}_0$. Thus, $\abs{X_{t} - X_{t}^*} \leq \max \left( \abs{X_{k+1} - X_{k+1}^*}, 2 \left( H + 3\abs{b} \right) D  \right)$.
    
    We have found that for all $d \in (0,\min(1,\abs{b}))$ and $k \in \mathbb{N}$, on the event $\{ T_d = k \}$,  $\abs{X_t - X_t^*} \leq (k+1) \abs{b}2D$ for all $t \leq k + 1$, and $\abs{X_{t} - X_{t}^*} \leq \max \left( \abs{X_{k+1} - X_{k+1}^*}, 2 \left( H + 3\abs{b} \right) D  \right)$ for all $t \geq k + 1$. Thus, we conclude that for all $d \in (0,\min(1,\abs{b}))$, $k \in \mathbb{N}$ and $t \in \mathbb{N}_0$, on the event $\{ T_d = k \}$, $\abs{X_{t} - X_{t}^*} \leq \max \left( (k + 1) \abs{b} 2D, 2 \left( H + 3\abs{b} \right) D  \right) = \max \left( (k + 1) \abs{b} 2D, 2 \left(\frac{\abs{b}+d}{1-d} + 3\abs{b} \right) D  \right)$.
\end{proof}

\begin{lemma} \label{lemma:similar-controller-behaviour-after-time-T}
Consider the parameter estimates $((\hat{a}_t,\hat{b}_t))_{t \in \mathbb{N}}$ from \eqref{eqn:parameter-estimate}. Suppose A1-A2 and $a \in \{-1,1\}$ hold on the closed-loop system \eqref{eqn:closed-loop-system-v2} and $x_0 \in \mathbb{R}$. Then, for all $d \in (0,\min(1, \abs{b}))$, $k \in \mathbb{N}$ and $t \geq k + 1$, $\hat{a}_{t-1}/\abs{\hat{a}_{t-1}} = a/\abs{a}$ and $\hat{b}_{t-1}/\abs{\hat{b}_{t-1}} = b/\abs{b}$ hold on the event $\{ T_d = k \}$ (with $T_d$ defined in \eqref{eqn:definition-Td}):
\end{lemma}
\begin{proof}
Suppose $k \in \mathbb{N}$, and $t \geq k + 1$. On the event $\{ T_d = k \}$, we have $\abs{\hat{a}_{t-1} - a} \leq \norm{\hat{\theta}_{t-1} - \theta_*}_{\infty} \leq \norm{\hat{\theta}_{t-1} - \theta_*}_2 < 1$, and so $\hat{a}_{t-1} \in (a - 1, a + 1)$, which implies $\hat{a}_{t-1}/\abs{\hat{a}_{t-1}} = a/\abs{a}$.

Additionally, on the event $\{ T_d = k \}$ we have $\abs{\hat{b}_{t-1} - b} \leq \norm{\hat{\theta}_{t-1} - \theta_*}_{\infty} \leq \norm{\hat{\theta}_{t-1} - \theta_*}_2 < \abs{b}$, and so $\hat{b}_{t-1} \in (b - \abs{b}, b + \abs{b})$, which implies $\hat{b}_{t-1}/\abs{\hat{b}_{t-1}} = b/\abs{b}$.
\end{proof}

\addtolength{\textheight}{-12cm}

\end{document}